\documentclass[letterpaper]{article}
\usepackage{preamble}
\usepackage{algorithm}
\usepackage[noend]{algorithmic}
\usepackage{thm-restate}
\usetikzlibrary{decorations.pathreplacing}
\usepackage{newfloat}
\usepackage{listings}
\usepackage{authblk}
\title{Fair Division Among Couples and Small Groups}
\author{Paul Gölz}
\author{Hannane Yaghoubizade}
\affil{Cornell University, School of Operations Research and Information Engineering}
\date{}

\newcommand{\efone}{EF1}
\newcommand{\ef}[1]{EF\ensuremath{#1}}
\newcommand{\prop}[1]{PROP\ensuremath{#1}}
\newcommand{\propone}{PROP1}
\newcommand{\proptwo}{PROP2}

\newcommand{\fpo}{fPO}

\begin{document}
\maketitle

\begin{abstract}  
We study the fair allocation of indivisible goods across groups of agents, where each agent fully enjoys all goods allocated to their group.
We focus on groups of two (\emph{couples}) and other groups of small size.
For two couples, an \efone{} allocation\,---\,one in which all agents find their group's bundle no worse than the other group's, up to one good\,---\,always exists and can be found efficiently.
For three or more couples, \efone{} allocations need not exist.

Turning to proportionality, we show that, whenever groups have size at most $k$, a \prop{k} allocation exists and can be found efficiently.
In fact, our algorithm additionally guarantees (fractional) Pareto optimality, and \propone{} to the first agent in each group, \proptwo{} to the second, etc., for an arbitrary agent ordering.
In special cases, we show that there are \propone{} allocations for any number of couples.
\end{abstract}

\section{Introduction}
Four siblings\,---\,Anna, Ben, Carmen, and Dave\,---\,jointly own a cottage on the coast and are currently deciding which sibling’s family will get to stay in the cottage in which weeks of the year.
Each sibling $i$ has a utility $u_i(\alpha) \geq 0$ for each week $\alpha$; for example, Anna prefers Spring weeks over Summer due to milder temperatures, and Ben would particularly value being at the cottage for July 4.
Assume (as we will throughout the paper) that a sibling's utility for a set $B$ of weeks is \emph{additive} (i.e., that $u_i(B) = \sum_{\alpha \in B} u_i(\alpha)$) and that the siblings are treating the weeks as \emph{indivisible} goods, i.e., they don’t want to allocate fractions of a week or assign a week to several families.

\looseness=-1
To solve the siblings' predicament as stated so far, the field of fair division offers allocation algorithms~\cite[e.g.,][]{LMM+04,CKM+19} with compelling axiomatic guarantees.
In particular, algorithms like \emph{envy-cycle elimination} and \emph{Maximum Nash welfare} ensure \emph{envy freeness up to one good (\efone)}.
\efone{} means that each sibling $i$ finds their assigned weeks $B_i$ to be at least as valuable as the weeks $B_j$ assigned to any other sibling, at least when removing some week $\alpha$ from $B_j$: $u_i(B_i) \geq \min_{\text{good}~\alpha} u_i(B_j \setminus \{\alpha\})$.\footnote{The strengthening of this axiom without the removal of a good, \emph{envy freeness}, is not always satisfiable for indivisible goods. E.g., if all siblings only have positive utility for a single week, the siblings who do not receive this week will always be envious.}
A second key axiom, \emph{proportionality up to one good (\propone{})}, states that each sibling $i$ receives at least their proportional share of their utility for all weeks, at least when adding some week $\alpha$: $\max_{\text{week}~\alpha} u_i(B_i \cup \{\alpha\}) \geq \frac{u_i(M)}{n}$, where $n=4$ is the number of siblings and $M$ the set of all weeks.

Our scenario deviates from the classic fair division setting in that the siblings' spouses also have utilities, which need not align with their partners'.
An allocation of weeks that Anna finds fair may still make her husband Alex envy another sibling or perceive their family's assignment as falling short of proportionality.
Is it possible to allocate the weeks over the families so that axioms such as \efone{} and \propone{} hold from the perspectives of all four siblings and their respective spouses?
Equivalent fair division problems may arise in splitting an inheritance between families or dissolving business partnerships, whenever the entities receiving allocations consist of two persons whose preferences should be satisfied.

The question we raise above fits into the model of group fair division~\cite{MS17a,KSV20}, but little is known for small groups such as couples.
\textcite{KSV20} show that \efone{} allocations need not exist for two groups with three members each and show that an (as of yet, unproven) graph conjecture by \textcite{JA17} would imply \efone{} existence for two couples.
With the basic question of \efone{} existence for two couples unresolved, Kyropoulou et al.\ instead focus on how much axioms like \efone{} must be relaxed to be satisfiable for two large groups.
Meanwhile, the question of whether \efone{} might exist for arbitrarily many couples has remained open.

In this paper, we aim to answer the following question:
\begin{quote}
\emph{For fair allocation over couples, and groups with few members more broadly, can we always guarantee the existence of \efone{} or \propone{} allocations? If no, can we guarantee slight approximations?}
\end{quote}

A second motivation is that this question touches on the fundamental combinatorics of fair allocations to \emph{individuals}.
Indeed, fix a number $n$ of agents and $m$ of indivisible goods, and consider the set $\mathcal{A}$ of all $n^m$ allocations of goods across agents $1$ through $n$.
Each vector of utility functions $u = (u_1, \dots, u_n)$ determines a subset $F_u \subseteq \mathcal{A}$ of allocations that are EF1 (or PROP1 etc.) for these utilities; call the family of all such sets $\mathcal{F}$.
An allocation to $n$ couples is EF1 iff it lies in $F_u$ for the utilities $u$ of the first partners in each couple and in $F_{u'}$ for the second partners' utilities $u'$.
That is, EF1 allocations always exist for couples iff any two sets $F_u \in \mathcal{F}$ intersect, i.e., if $\mathcal{F}$ is an \emph{intersecting family}~\cite{Jukna11}.

\subsection{Our Techniques and Results}
We begin by proving in \cref{sec:ef1:twocouples} that EF1 allocations always exist for two couples, proving a case left open by \textcite{KSV20} without relying on the graph conjecture mentioned above.
Our proof takes a different path, by rounding a fractional allocation computed via an unconventional linear program, which also yields a polynomial-time algorithm.
Though there are only a few ways of rounding this LP, our proof that one of these roundings must satisfy EF1 requires a non-trivial combinatorial argument.
In \cref{sec:ef1:morecouples}, we prove that EF1 cannot be guaranteed for three (or more) couples.

Since EF1 is not achievable for more than two couples, we turn our focus on PROP1 and its variants in \cref{sec:prop}.
In the setting of arbitrarily many couples, we can guarantee the slightly weaker guarantee of PROP2, in addition to fractional Pareto optimality, an axiom of allocation efficiency.
This follows from our main result\,---\,an efficient, iterative-rounding algorithm that works for groups of arbitrary sizes, and guarantees PROP1 to the first member of each group, PROP2 to the second, etc., according to an arbitrary ordering of group members.
In various special cases, we can show that PROP1 allocations exist for arbitrarily many couples, for example if utilities are dichotomous and all agents value the same number of goods. PROP1 allocations need not exist for groups of three agents.

In \cref{sec:empirics}, we study fair division among couples empirically, by taking utilities from real-world allocation problems submitted to Spliddit~\cite{GP14a} and pairing up the agents.
EF1 and PROP1 allocations exist for every single fair division problem we study, suggesting that they are ubiquitous in practice.
We also study the iterative rounding algorithm and find that it almost always provides PROP1 for all agents, and even EF1 in most cases.

\subsection{Related Work}
Fair division among groups was independently introduced by \textcite{MS17a}, and \textcite{SN2019}, the former studying indivisible goods and the latter divisible goods.
This line of work has since been expanded, with several works exploring fair allocation in both divisible and indivisible settings~\cite{SUKSOMPONG2018, ghodsi2018, KSV20, SS2020, MS2022, SS2023, allocator2023, CKS2025, MM2025}. In contrast to our work, most research on group fair division with indivisible goods focuses on asymptotic analysis for large groups.

The paper by \textcite{allocator} is closely related to our work.
While their motivation for studying two sets of utilities does not make reference to groups, their setting is equivalent to fair division among couples.
Concurrently to us,\footnote{An earlier version~\cite{allocator2023} relied on the conjecture by \textcite{JA17} to prove \ef{1} for two couples. The very recently revised preprint~\cite{allocator} removes this assumption and includes an algorithm for additive utilities.} \textcite{allocator} also show that \ef{1} allocations exist for two couples.
Their argument makes heavy use of combinatorial results to show \ef{1} existence even for general monotone valuations.

They study proportionality as well, proving the existence of
\prop{\text{-}O(\log(n))} allocations for additive functions, where $n$ is the number of couples. Our iterative rounding algorithm improves the gap in proportionality to a constant (\prop{2}) and naturally extends to arbitrary group sizes.

\section{Preliminaries}
A group fair division instance consists of a set of goods $M = [m] =\{1, \dots, m\}$, a set $\mathcal{G}$ of $n \geq 2$ groups of agents, and the agents' valuations.
For a group $g \in \mathcal{G}$, we use $|g|$ to denote the number of agents in that group. We refer to the $i$-th agent in $g$ as $(g, i)$ for $1 \leq i \leq |g|$.
Each agent $(g, i)$ has a value $u_{gi}(\alpha)$ for each good $\alpha \in M$, which induce an additive valuation function $u_{gi}(B) \coloneqq \sum_{\alpha \in B} u_{gi}(\alpha)$ over sets of goods $B$. 
We say that the agent's valuation is \emph{binary} if $u_{gi}(\alpha)\in \{0,1\}$ for all $\alpha \in M$.

A \emph{fractional allocation} is a vector \( x \in [0,1]^{M \times \mathcal{G}} \), where \( x_{\alpha g} \) denotes the fraction of good \( \alpha \) assigned to group \( g \), such that \( \sum_{g \in \mathcal{G}} x_{\alpha g} = 1 \) for each \( \alpha \in M \).
If $x \in \{0,1\}^{M \times \mathcal{G}}$, we call it a \emph{(discrete) allocation}.
Equivalently, we represent an allocation as a partition of the goods into \emph{bundles} $\{B_g\}_{g \in \mathcal{G}}$, where $B_g$ is the bundle group $g$ receives.
An allocation is \emph{balanced} if $\big||B_g| - |B_{g'}|\big|\leq 1$ for all $g, g'\in \mathcal{G}$.
Since all agents in $g$ fully enjoy the goods in their bundle $B_g$, agent $(g,i)$'s utility for such an allocation is $u_{gi}(B_g)$.
We linearly extend utilities to fractional allocations so that $(g, i)$'s utility for allocation $x$ is $\sum_{\alpha \in M} x_{\alpha, g} u_i(\alpha)$.

    For any $k \geq 0$, an allocation $\{B_g\}_{g \in \mathcal{G}}$ is
    \begin{itemize}
        \item \emph{envy-free up to $k$ goods (\ef{k})} for an agent $(g, i)$ if, for any $g' \in \mathcal{G}$, there is a set $B \subseteq B_{g'}$ such that $|B| \leq k$ and $u_{gi}(B_g) \geq u_{gi}(B_{g'}\setminus B)$.
        \item \emph{proportional up to $k$ goods (\prop{k})} for an agent $(g, i)$ if there exists a set $B \subseteq M \setminus B_{g}$ such that $|B| \leq k$ and $u_{gi}(B_g \cup B) \geq u_{gi}(M) / n$.
    \end{itemize}
An allocation is \emph{envy-free (\ef{})} for an agent if it is \ef{}0 and \emph{proportional (\prop{})} if its \prop{}0.
We say an allocation is \ef{k} if it is \ef{k} for every agent, and analogously for \ef{}, \prop{k}, and \prop{}. \ef{} and \prop{} also naturally extend to fractional allocations.

A fractional allocation $x$ \emph{Pareto dominates} another fractional allocation $x'$ if, for all agents $(g, i)$, $u_{gi}(x) \geq u_{gi}(x')$ and if this inequality is strict for at least one agent.
A fractional allocation $x$ is \emph{fractionally Pareto optimal (fPO)} if it is not Pareto dominated by any other fractional allocation.
Note that, for a discrete allocation, being fPO implies the more classic axiom of Pareto optimality (i.e., not being Pareto dominated by any discrete allocation).

\paragraph{Linear Programming.}
We recall some notions from linear programming.
A set $P = \{x \in \mathbb{R}^n: Ax \leq b\}$ for some $A\in \mathbb{R}^{m \times n}$, and $b \in \mathbb{R}^m$ is a \emph{polyhedron}, and a bounded polyhedron is called a \emph{polytope}.
A point $z \in P$ is called a \emph{Basic Feasible Solution (BFS)} if there are $n$ linearly independent rows of $A$ such that $Az \leq b$ holds with equality in these $n$ rows.
\emph{Linear Programming (LP)} consists of maximizing (or minimizing) a linear function over a polyhedron, i.e., $\max\{c^Tx: Ax \leq b, x \in \mathbb{R}^n\}$. 
We repeatedly use:
\begin{proposition}[{{\cite[Thm.\ 2.8]{BT97}}}]
\label{prop:optbfs}
    If $P$ is a non-empty polytope, there exists a BFS that achieves $\max\{c^Tx: x\in P\}$ for a given $c\in \mathbb{R}^n$.
\end{proposition}
Such an optimal BFS can be found in (weakly) polynomial time using the ellipsoid method~\cite{khachiyan1980polynomial}.

\section{\efone{} Among Couples}
\label{sec:ef1}
In this section, we study \efone{} allocations for \emph{couples}, i.e., groups $g$ that all have size $2$.
For $n=2$ couples, we show that the ``gold standard''~\cite{FSV+19} of \efone{} can be achieved even when we must satisfy twice as many agents per bundle, compared to the classic, individual setting.
However, this positive result does not extend further, as we prove that \efone{} allocations may not exist for $n \geq 3$ couples.

\subsection{Existence of \efone{} with Two Couples}
\label{sec:ef1:twocouples}
In this part, we consider the case with two groups, which we call the \emph{first} group $f$ and the \emph{second} group $s$. 
\textcite{KSV20} proved that a balanced \efone{} allocation always exists when $(|f|, |s|) = (2, 1)$, and left the existence of such an allocation for $(|f|, |s|) = (2, 2)$ as an open question, which we answer in the affirmative.

Our high-level approach is to round an appropriate fractional allocation into an \efone{} (discrete) allocation.
Starting from a fractional allocation is promising because envy freeness is always achievable in this domain (say, by splitting each good equally between groups).
Broadly speaking, fractional allocations $x$ are easiest to round if they are already ``almost discrete'' (i.e., most entries are $0$ or $1$).
In such cases, we only need to round the few remaining fractional entries to 0 and 1, which yields a limited number of discrete allocations to reason over, all of which are still close to $x$ and might therefore be ``almost'' envy-free.

The most direct attempt at pursuing this approach would be to round a BFS from the polytope of envy-free allocations, which is defined as follows, where $x_{\alpha}$ is the share of good $\alpha$ given to the first group:
{
\begin{align*}
    &\textstyle\sum_{\alpha \in M} x_{\alpha} \, u_{fi}(\alpha) \geq \sum_{\alpha \in M} (1-x_{\alpha}) \, u_{fi}(\alpha) &i=1,2& \\
     &\textstyle\sum_{\alpha \in M} (1-x_{\alpha}) \, u_{si}(\alpha) \geq \sum_{\alpha \in M} x_{\alpha} \, u_{si}(\alpha) &i=1,2& \\
    &0 \leq x_{\alpha} \leq 1 &\forall \alpha \in M&.
\end{align*}
}

Using a BFS of this polytope allows us to obtain an almost discrete allocation.
Indeed, since there are $m$ variables, $m$ constraints must be tight in a BFS, and hence at least $m-4$ constraints of the shape $0 \leq x_\alpha$ or $x_\alpha \leq 1$ are tight.
This implies that at most four goods $\alpha$ are allocated \emph{fractionally} (i.e., $0 < x_\alpha < 1$), leaving $2^4 = 16$ ways of rounding.

Unfortunately, a BFS of this polytope may not have a way to be rounded into an \efone{} allocation. For example, consider the following valuations over goods $\{1, 2, 3, 4\}$:

\begin{center}
        \begin{tabular}{ccccc}
            \toprule
             valuation & $1$ & $2$ & $3$ & $4$ \\
             \midrule
             $u_{f 1}$ & $1$ & $0$ & $0.1$ & $0.1$\\
             $u_{f 2}$ & $0$ & $1$ & $0.1$ & $0.1$ \\
             $u_{s 1}$ & $0.5$ & $0.5$ & $0.1$ & $0.1$ \\
             $u_{s 2}$ & $0.2$ & $0$ & $0.5$ & $0.5$ \\
             \bottomrule
        \end{tabular}
\end{center}

In this case, one BFS\footnote{Specifically, the EF fractional allocation with maximum utilitarian welfare.} allocates a $3/5$ fraction of goods $1$ and $2$ to group $f$, and the rest of $1$ and $2$ plus the entirety of $3$ and $4$ to group $s$.
But no way of rounding the fractional goods leaves all agents \efone{}: agent $(f, 1)$ requires good $1$ to be given to $f$, $(f, 2)$ requires the same for good $2$, but giving both to $f$ leaves $(s, 1)$ envious.

To obtain a working rounding argument, we devise an alternative polytope whose BFS has even fewer fractional variables, and whose roundings are all \efone{} for agent $(f, 1)$, leaving us with one fewer agent to worry about. 
For this, assume w.l.o.g.\ that the number of goods $m$ is even (otherwise, we  add a dummy good with value 0 for every agent) and that the goods are ordered according to $(f, 1)$'s valuation, i.e., $u_{f1}(1) \geq \dots \geq u_{f1}(m)$.

We restrict ourselves to allocations in which each group receives exactly one out of the goods $\{1, 2\}$, one out of $\{3, 4\}$, \dots, and one out of $\{m-1, m\}$.
This structure is inspired by \textcite{KSV20}, who observe that any allocation with this structure is \efone{} for $(f, 1)$, 
allowing us to focus on satisfying \efone{} for the remaining three agents.
The structure also ensures balancedness.

This structure naturally generalizes to fractional allocations by ensuring that each group receives a total of one unit from each pair \( \{2j-1, 2j\} \) (for $1 \leq j \leq m/2$).
Specifically, if the first group receives a fraction \( y_j \) of good \( 2j-1 \), it receives \( 1 - y_j \) of good \( 2j \).

The polytope of fractional allocations of this structure, which are furthermore \ef{} for the three remaining agents, can be written as follows:
{
\footnotesize
\begin{align*}
        &\sum_{\mathclap{j\in [m/2]}} (2y_j\!-\!1) \, u_{f2}(2j\!-\!1) + (1\!-\!2y_j) \, u_{f2}(2j)\geq 0\\
        &\sum_{\mathclap{j\in [m/2]}} (1\!-\!2y_j) \, u_{si}(2j\!-\!1) + (2y_j\!-\!1) \, u_{si}(2j) \geq 0 & i=1, 2\\
        &0 \leq y_j \leq 1 & \mathllap{j = 1, \dots, m/2.}
\end{align*}
}

A BFS of this polytope has at most three fractional values.
We avoid one more fractional value and simplify the argument by not only requiring \ef{}, but maximizing the minimum gap $d$ by which any agent prefers their bundle over the other.
We believe that this trick for eliminating one more fractional variable can be useful for other settings as well.
{
\footnotesize
\begin{align*}
    \text{max}~&d \\
        \text{s.t.}~&\sum_{\mathclap{j\in [m/2]}} (2y_j\!-\!1) \, u_{f2}(2j\!-\!1) + (1\!-\!2y_j) \, u_{f2}(2j)\geq \mathrlap{d}\\
        &\sum_{\mathclap{j\in [m/2]}} (1\!-\!2y_j) \, u_{si}(2j\!-\!1) + (2y_j\!-\!1) \, u_{si}(2j) \geq d & i=1, 2\\
        &0 \leq y_j \leq 1 & \mathllap{j = 1, \dots, m/2.}
\end{align*}
}

Because this formulation has one more variable, one more constraint is binding at a BFS, which ensures that there are at most two fractional $y_j$.
Let $(y^*, d^*)$ be an optimal BFS for the LP, which can be found efficiently (\cref{prop:optbfs}). 
Since setting all $y_j = 1/2$ and $d = 0$ is feasible, we know that $d^*$ is nonnegative and that $y^*$ describes a fractional allocation that is \ef{} for the three agents.

\definecolor{mygreen}{RGB}{77, 233, 76}
\definecolor{myred}{RGB}{246, 0, 0}
\definecolor{myblue}{RGB}{55, 131, 255}
\definecolor{myorange}{RGB}{255, 140, 0}
\definecolor{lightgray}{RGB}{220, 220, 220}
\definecolor{darkgray}{RGB}{123,112,96}
\definecolor{mybowcolor}{RGB}{212, 210, 173}

\tikzset{
    pics/gift/.style={
        code = {
            \tikzset{gift/.cd, #1}
            \def\pv##1{\pgfkeysvalueof{/tikz/gift/##1}}
            \begin{scope}[scale=0.25]
                \pgfmathsetmacro{\splity}{-2.5 + 3 * \pv{split}}
                \fill[\pv{boxcolor}, opacity=\pv{bottomopacity}] (-1.8, -2.5) rectangle (1.8, \splity);
                \fill[\pv{boxcolor}, opacity=\pv{topopacity}] (-1.8, \splity) rectangle (1.8, 0.5);
                \fill[\pv{bowcolor}!90!brown] (0, 0.5) ellipse [x radius=0.4, y radius=0.2];
                \fill[\pv{bowcolor}!80!yellow] (0, 0.4) .. controls (-2, 2.4) and (-1.8, -0.4) .. (-0.5, 0.4) .. controls (-1.1, 0.8) and (-0.6, 1.4) .. (0, 0.4);
                \draw[mybowcolor] (-1.8, -2.5) rectangle (1.8, 0.5);
                \draw[very thick, \pv{bowcolor}!90!brown] (0, 0.4) .. controls (-2, 2.4) and (-1.8, -0.4) .. (-0.5, 0.4) .. controls (-1.1, 0.8) and (-0.6, 1.4) .. (0, 0.4);
                \fill[\pv{bowcolor}!80!yellow] (0, 0.4) .. controls (2, 2.4) and (1.8, -0.4) .. (0.5, 0.4) .. controls (1.1, 0.8) and (0.6, 1.4) .. (0, 0.4);
                \draw[very thick, \pv{bowcolor}!90!brown] (0, 0.4) .. controls (2, 2.4) and (1.8, -0.4) .. (0.5, 0.4) .. controls (1.1, 0.8) and (0.6, 1.4) .. (0, 0.4);
                \fill[\pv{bowcolor}!80!yellow] (0, 0.5) .. controls (-0.5, 0.2) and (-0.7, -0.5) .. (-0.5, -1) .. controls (-0.4, -0.6) and (-0.3, -0.1) .. (0, 0.5);
                \draw[very thick, \pv{bowcolor}!90!brown] (0, 0.5) .. controls (-0.5, 0.2) and (-0.7, -0.5) .. (-0.5, -1) .. controls (-0.4, -0.6) and (-0.3, -0.1) .. (0, 0.5);
                \fill[\pv{bowcolor}!80!yellow] (0, 0.5) .. controls (0.5, 0.2) and (0.7, -0.5) .. (0.5, -1) .. controls (0.4, -0.6) and (0.3, -0.1) .. (0, 0.5);
                \draw[very thick, \pv{bowcolor}!90!brown] (0, 0.5) .. controls (0.5, 0.2) and (0.7, -0.5) .. (0.5, -1) .. controls (0.4, -0.6) and (0.3, -0.1) .. (0, 0.5);
            \end{scope}
        }
    },
    gift/.cd,
    boxcolor/.initial = red!30,
    bowcolor/.initial = orange,
    split/.initial = 1,
    bottomopacity/.initial = 1,
    topopacity/.initial = 1,
}

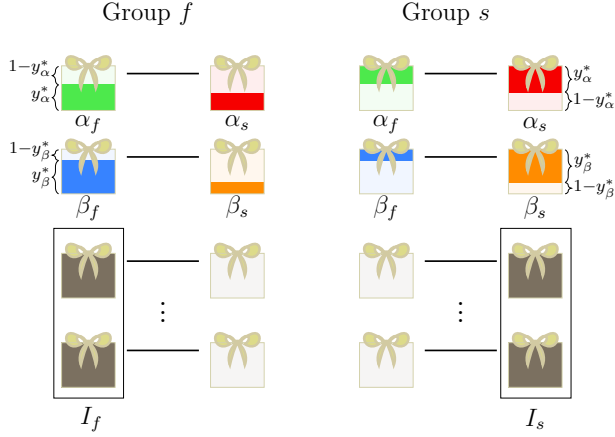
\begin{figure}[tb]
    \centering
    \resizebox{0.5\columnwidth}{!}{%
\begin{tikzpicture}[
    node distance=0.4cm and 0.8cm,
    every node/.style={font=\large}
]

\begin{scope}[shift={(-7,0)}]
    \node at (0, 4.2) {\large Group $f$};
    \node (alpha_f_top) at (-1, 2.35) {$\alpha_f$\vphantom{$\alpha_{fs}$}};
    \node (box_green) at (-1, 3.2) {\tikz \pic{gift={boxcolor=mygreen, bowcolor=mybowcolor, split=0.6, topopacity=0.07}};};
    \draw[decorate, decoration={brace, amplitude=3pt}]  (-1.5, 3.01) -- (-1.5, 3.32);
    \node[left] at (-1.5,3.3) {\footnotesize{$1{-}y^*_{\alpha}$}};
    \node (box_red) at (1.5, 3.2) {\tikz \pic{gift={boxcolor=myred, bowcolor=mybowcolor, split=0.4, topopacity=0.07}};};
    \node (alpha_s_top) at (1.5, 2.35) {$\alpha_s$\vphantom{$\alpha_{fs}$}};
    \draw[decorate, decoration={brace, amplitude=3pt}]  (-1.5, 2.58) -- (-1.5, 3.01);
    \node[left] at (-1.5,2.85) {\footnotesize{$y^*_{\alpha}$}};
    \draw[thick] (box_red) -- (box_green);
    \node (beta_f_bottom) at (-1, 0.95) {$\beta_f$\vphantom{$\beta_{fs}$}};
    \node (box_blue) at (-1, 1.8) {\tikz \pic{gift={boxcolor=myblue, bowcolor=mybowcolor, split=0.75, topopacity=0.07}};};
    \draw[decorate, decoration={brace, amplitude=3pt}]  (-1.5, 1.74) -- (-1.5, 1.92);
    \node[left] at (-1.5,1.9) {\footnotesize{$1{-}y^*_{\beta}$}};
    \node (box_orange) at (1.5, 1.8) {\tikz \pic{gift={boxcolor=myorange, bowcolor=mybowcolor, split=0.25, topopacity=0.07}};};
    \node (beta_s_bottom) at (1.5, 0.95) {$\beta_s$\vphantom{$\beta_{fs}$}};
    \draw[decorate, decoration={brace, amplitude=3pt}]  (-1.5, 1.18) -- (-1.5, 1.74);
    \node[left] at (-1.5,1.45) {\footnotesize{$y^*_{\beta}$}};
    \draw[thick] (box_blue) -- (box_orange);
    \begin{scope}[yshift=0.8cm]
        \draw (-1.58, -3.05) rectangle (-0.41, -0.2);
        \node at (-0.95, -3.4) {$I_f$};
        \node (l1) at (-1, -0.75) {\tikz \pic{gift={boxcolor=darkgray, bowcolor=mybowcolor}};};
        \node (r1) at (1.5, -0.75) {\tikz \pic{gift={boxcolor=darkgray, bowcolor=mybowcolor, bottomopacity=0.07, topopacity=0.07}};};
        \draw[thick] (l1) -- (r1);
        \node (l2) at (-1, -2.25) {\tikz \pic{gift={boxcolor=darkgray, bowcolor=mybowcolor}};};
        \node (r2) at (1.5, -2.25) {\tikz \pic{gift={boxcolor=darkgray, bowcolor=mybowcolor, bottomopacity=0.07, topopacity=0.07}};};
        \draw[thick] (l2) -- (r2);
        \node at (0.25, -1.5) {\LARGE{$\vdots$}};
    \end{scope}
\end{scope}

\begin{scope}[shift={(-2,0)}]
    \node at (0, 4.2) {\large Group $s$};
    \node (alpha_f_top_s) at (-1, 2.35) {$\alpha_f$};
    \node (box_green_s) at (-1, 3.2) {\tikz \pic{gift={boxcolor=mygreen, bowcolor=mybowcolor, split=0.6, bottomopacity=0.07}};};
    \node (box_red_s) at (1.5, 3.2) {\tikz \pic{gift={boxcolor=myred, bowcolor=mybowcolor, split=0.4, bottomopacity=0.07}};};
    \node (alpha_s_top_s) at (1.5, 2.35) {$\alpha_s$};
    \draw[decorate, decoration={brace, amplitude=3pt}]  (2, 3.32) -- (2, 2.89);
    \node[left] at (2.6,3.2) {\footnotesize{$y^*_{\alpha}$}};
    \draw[thick] (box_red_s) -- (box_green_s);
    \draw[decorate, decoration={brace, amplitude=3pt}]  (2, 2.89) -- (2, 2.58);
    \node[left] at (2.97,2.75) {\footnotesize{$1{-}y^*_{\alpha}$}};
    \node (beta_f_bottom_s) at (-1, 0.95) {$\beta_f$};
    \node (box_blue_s) at (-1, 1.8) {\tikz \pic{gift={boxcolor=myblue, bowcolor=mybowcolor, split=0.75, bottomopacity=0.07}};};
    \draw[decorate, decoration={brace, amplitude=3pt}]  (2, 1.92) -- (2, 1.37);
    \node[left] at (2.6,1.7) {\footnotesize{$y^*_{\beta}$}};
    \node (box_orange_s) at (1.5, 1.8) {\tikz \pic{gift={boxcolor=myorange, bowcolor=mybowcolor, split=0.25, bottomopacity=0.07}};};
    \node (beta_s_bottom_s) at (1.5, 0.95) {$\beta_s$};
    \draw[thick] (box_blue_s) -- (box_orange_s);
    \draw[decorate, decoration={brace, amplitude=3pt}]  (2, 1.37) -- (2, 1.18);
    \node[left] at (2.97,1.25) {\footnotesize{$1{-}y^*_{\beta}$}};
    \begin{scope}[yshift=0.8cm]
        \draw (0.92, -3.05) rectangle (2.1, -0.2);
        \node at (1.5, -3.4) {$I_s$};
        \node (l1_s) at (-1, -0.75) {\tikz \pic{gift={boxcolor=darkgray, bowcolor=mybowcolor, bottomopacity=0.07, topopacity=0.07}};};
        \node (r1_s) at (1.5, -0.75) {\tikz \pic{gift={boxcolor=darkgray, bowcolor=mybowcolor}};};
        \draw[thick] (l1_s) -- (r1_s);
        \node (l2_s) at (-1, -2.25) {\tikz \pic{gift={boxcolor=darkgray, bowcolor=mybowcolor, bottomopacity=0.07, topopacity=0.07}};};
        \node (r2_s) at (1.5, -2.25) {\tikz \pic{gift={boxcolor=darkgray, bowcolor=mybowcolor}};};
        \draw[thick] (l2_s) -- (r2_s);
        \node at (0.25, -1.5) {\LARGE{$\vdots$}};
    \end{scope}
\end{scope}
\end{tikzpicture}%
} 
    \caption{Illustration of $y^*$. Each good is paired with the good next to it. The bright, solid portion of a box represents the fraction of the good a group receives, while the faded portion represents the fraction allocated to the other group.}\label{fig:LPd}
\end{figure}

We are now ready to prove \efone{} existence, by rounding the fractional solution $y^*$ into an \efone{} allocation, in which each group receives exactly one good among $\{2j-1, 2j\}$ for each $1 \leq j \leq m/2$. 
Since the rounding argument is fairly complex, we outline the main ideas here and defer the proof to \cref{thm:2v2proof}.
\begin{restatable}{theorem}{twovstwo}
\label{thm:2v2}
In the case of two groups with two agents each, a balanced \efone{} allocation always exists and can be found in (weakly) polynomial time.
\end{restatable}
\begin{proof}[Proof sketch]
Since $y^*$ has at most two fractional values, we can fix $1 \leq \alpha, \beta \leq m/2$ such that all variables except for $y^*_\alpha, y^*_\beta$ are integral. Let $I_f$ and $I_s$ be the set of remaining goods that are allocated entirely to the first and second group, respectively.
We can assume w.l.o.g.\ that $y^*_\alpha, y^*_\beta \geq 1/2$.\footnote{Otherwise, one can swap the roles of, say, goods $2 \alpha - 1$ and $2 \alpha$, which keeps $(f,1)$ \efone{}.}
For convenience, set $\alpha_f \coloneqq 2\alpha-1$ to be the good of which group $f$ receives a $y^*_\alpha$ fraction and group $s$ receives a $1 - y^*_\alpha$ fraction; $\alpha_s\coloneqq 2\alpha$ to be the good of which $s$ receives a $y^*_\alpha$ fraction and $f$ receives a $1-y^*_\alpha$ fraction; and analogously for $\beta_f, \beta_s$.
This allocation is illustrated in \cref{fig:LPd}.

\paragraph{Case $y^*_\alpha + y^*_\beta \geq 3/2$.} If $y^*_\alpha + y^*_\beta \geq 3/2$, the fractional allocation is already very close to being integral.
In this case, allocating $\{\alpha_f, \beta_f\}$ to $f$ and $\{\alpha_s, \beta_s\}$ to $s$ turns out to be \efone{}.
Since this allocation gives each good to the group that had the larger fraction of it in $y^*$, we refer to this as the \emph{natural rounding}.
To see that the natural rounding is \efone{}, observe that the natural rounding is reached by starting from $y^*$, and transferring goods as follows:
\begin{center}
\begin{tikzpicture}
\node (A) at (0,0) {group $f$};
\node (B) at (6.5,0) {group $s$};

\draw[->, thick] ([yshift=2mm]A.east) -- ([yshift=1.5mm]B.west) node[midway, above] {\footnotesize $(1\!-\!y^*_\alpha) \times \alpha_s, (1\!-\!y^*_\beta) \times \beta_s$};

\draw[->, thick] ([yshift=-2mm]B.west) -- ([yshift=-1.5mm]A.east) node[midway, below] {\footnotesize $(1\!-\!y^*_\alpha) \times \alpha_f, (1\!-\!y^*_\beta) \times \beta_f$};
\end{tikzpicture}
\end{center}
Taking the perspective of, say, $(f, 1)$, they start from the envy-free allocation $y^*$ and receive some fraction of $\alpha_f$ and $\beta_f$, which only reduces their envy.
Then, $f$ hands some fraction of $\alpha_s, \beta_s$ to $s$, but the the amount of this transfer is $1-y^*_\alpha + 1 - y^*_\beta \leq 1/2$ goods.
This transfer increase $(f, 1)$'s envy twofold because $f$'s allocation shrinks and that of $s$ grows.
But $(f, 1)$'s envy is now at most $\max(u_{f1}(\alpha_s), u_{f1}(\beta_s))$, which can be eliminated by removing the higher-valued good from $s$'s bundle.

\paragraph{Case $y^*_\alpha + y^*_\beta < 3/2$.} In the remaining case, in which $3/2 > y^*_\alpha + y^*_\beta \geq 1$, the fractional allocation is further from the natural rounding.
As a consequence, we have to reason about which of the four rounding options (in which $f$ receives $\{\alpha_f, \beta_f\}$, $\{\alpha_f, \beta_s\}$, $\{\alpha_s, \beta_f\}$, or $\{\alpha_s, \beta_s\}$, respectively) are \efone{} for each agent to find a rounding option that works for everyone.
We call agent $(g, i)$ \emph{unhappy} with $\alpha_g$ if they prefer $\alpha_{g'}$ over it (where $g'$ is the other group) and unhappy with $\beta_g$ if they prefer $\beta_{g'}$.
The following two observations follow from arguments similar to the one of the previous case:
\begin{enumerate}[leftmargin=*,labelindent=.65em]
    \item[(A)] If an agent $(g,i)$ is unhappy with both $\alpha_g$ and $\beta_g$, they are \efone{} for all rounding options except the natural one.
    \item[(B)] If an agent is happy with at least one of $\alpha_g$ or $\beta_g$, they are \efone{} for the natural rounding and at least one other rounding option.
\end{enumerate}
Since any two agents have at least one \efone{} option in common, we successfully find an \efone{} allocation whenever at least one of the three agents is \efone{} for all four rounding options.

If, finally, all three agents have some rounding option that is not \efone{}, none of the four rounding options discussed so far may work, but we have one more ace up our sleeve:
\begin{enumerate}[leftmargin=*,labelindent=.65em]
    \item[(C)] If an agent is not \efone{} under some rounding option, then they become \efone{} under the other three options if we \emph{swap the integral parts} $I_f$ and $I_s$.
\end{enumerate}
Since this observation applies to all three agents, each of them rules out at most one of the four rounding options after swapping the integral parts, which leaves one that is \efone{} for all of them.\footnote{For the straight-forward EF polytope, swapping integral allocation parts does not overcome the rounding counterexample.}
Since this allocation is also still \efone{} for the set-aside agent $(f, 1)$, this establishes the claim.
\end{proof}

As the counting arguments do not refer to which group each agent is in, the same argument also shows the existence of \efone{} allocations for two groups of sizes $(|f|, |s|) = (3, 1)$, left open by \textcite{KSV20}, for the natural adaptation of the LP. We conclude:
\begin{corollary}
   When there are two groups with a total of four agents, a balanced \efone{} allocation exists and can be computed in (weakly) polynomial time.
\end{corollary}

\subsection{\efone{} Impossibility for Three or More Couples}
\label{sec:ef1:morecouples}
Since we were able to guarantee \efone{} existence for two couples, one may hope that this existence extends to any number of couples, an audacious hope that has not been contradicted by earlier papers.
In the special case where the agents $(g, 1)$ across all groups $g$ have identical valuations, \textcite{allocator2023} show that \efone{} allocations do exist, using a variant of envy-cycle elimination due to \textcite{cardinality}.

In general, however, we find that \efone{} allocation need no longer exist for three or more couples:

\begin{restatable}{theorem}{efonecex}
\label{lem:more couples}
For $n \geq 3$ couples, some fair division instances have no \efone{} allocations.
\end{restatable}
\begin{proof}
We prove the claim here for $n=3$ and generalize to $n>3$ in \cref{app:cex ef1}. Consider the following instance with 3 couples called $f, s, t$ and goods $\{1, 2, 3, 4, 5\}$:
\begin{center}
        \begin{tabular}{cccccc}
            \toprule
             valuation & $1$ & $2$ & $3$ & $4$ & $5$ \\
             \midrule
             $u_{f 1}$ & $2$ & $2$ & $0$ & $0$ & $1$\\
             $u_{f 2}$ & $0$ & $0$ & $2$ & $2$ & $1$ \\
             $u_{s 1}$ & $0$ & $2$ & $0$ & $2$ & $1$ \\
             $u_{s 2}$ & $2$ & $0$ & $2$ & $0$ & $1$ \\
             $u_{t 1}$ & $2$ & $0$ & $0$ & $2$ & $1$ \\
             $u_{t 2}$ & $0$ & $2$ & $2$ & $0$ & $1$ \\
             \bottomrule
        \end{tabular}
\end{center}

        Each agent has positive valuation for three goods: two with value $2$ and one with value $1$.
        The agent must receive at least one such good since, otherwise, some other group receives two or more of those goods, violating \efone.
        
            For the sake of contradiction, suppose that an \efone{} allocation exists. Since there are five goods and three groups, one group must receive a single good. Let this be the case for group $f$, w.l.o.g.\ by symmetry. Since this good must have positive value for both agents in the group, it must be good $5$.
            
           Because the remaining four goods are never liked by both agents in a group, the other two groups must receive two goods each.
           Consider the two goods given to group $s$.
           By construction, there must be some agent $(g, i)$ (not necessarily in group $s$) for whom both of these goods have value $2$, and for whose partner $(g, i')$ both goods have value $0$. If $g=s$, then $(g, i')$ receives value $0$ and must be envious. Otherwise, 
           $(g, i)$ envies group $s$ by more than one good.
        \end{proof}

\section{Proportionality}
\label{sec:prop}
Because an \efone{} allocation among couples may not exist, it is natural to ask whether the weaker axiom of \propone{} can be guaranteed instead.
For $n$ couples, \textcite{allocator2023} establish the existence of a \prop{\text{-}O(\log n)} allocation by iteratively bi-partitioning the agents and applying, in each step, a rounding argument for a fair allocation among two groups.\footnote{This argument is much easier than our argument for \efone{}. Though \ef{} and \prop{} are equivalent for two groups, \propone{} is weaker and easier to achieve through rounding than \efone{}.} 
In this section, we get much closer to the standard axiom of \propone{}; for $n$ couples, we achieve \propone{} for the first agent in each group and \proptwo{} for each second agent.

\subsection{Almost \prop{} Allocations for Small 
Groups}
\label{sec:prop:iterative}
In fact, the claim for couples follows from a general result for groups of arbitrary sizes, which shows the existence of an \fpo{} allocation in which every agent $(g,i)$ is \prop{i}.

We prove this existence using an algorithm based on the \emph{iterative rounding} method~\cite{Jain01}. This method has been widely used in combinatorial optimization, including fair allocation~\cite{CCK2009, NGUYEN2016, CMV2025}. Our rounding is inspired by the algorithm of \textcite{GAP} for the Generalized Assignment Problem.

Our algorithm maintains a sequence of fractional allocations, and iteratively freezes coordinates at $0$ and $1$ until it reaches a discrete allocation.

The steps of the algorithm are most easily explained by considering a bipartite graph, whose nodes on one side are the goods $M' \subseteq M$ not yet discretely allocated and on the other side are a set of $n$ groups $\mathcal{G}'$, obtained from $\mathcal{G}$ by removing some agents.
The set of edges $E \subseteq M' \times \mathcal{G}'$ denotes the \emph{allowed assignments} by specifying, for each good, the groups that the good may still be allocated to.
Initially, we have $M' \coloneqq M, \mathcal{G}' \coloneqq \mathcal{G}, E \coloneqq M \times \mathcal{G}$, meaning that no goods have been allocated, no agents eliminated, and all possible assignments are allowed.
Based on $M', \mathcal{G}', E$, and the bundles $B_g$ of goods already discretely allocated to group $g$, we consider the following polytope, which describes the currently allowed fractional allocations that are \prop{} for all agents remaining in $\mathcal{G}'$:
{
\footnotesize
\begin{align*}
        &\sum_{\mathclap{\alpha:  (\alpha, g) \in E}} x_{\alpha g} \, u_{gi}(\alpha) \geq \frac{u_{gi}(M)}{n}-u_{gi}(B_g) &\forall g\in \mathcal{G}', i \in [|g|]\notag \\
        &\sum_{\mathclap{g:(\alpha, g) \in E}} x_{\alpha g} = 1 & \forall \alpha \in M'\notag\\
        &0 \leq x_{\alpha g} \leq 1 & \forall (\alpha, g) \in E. 
\end{align*}}

We refer to the three types of constraints, in order from top to bottom,  as \emph{agent} constraints, \emph{good} constraints, and \emph{edge} constraints. Our algorithm is based on the following lemma, which will guide the rounding procedure.

\begin{lemma}\label{lemma:prop-cond}
    If $M'\neq \emptyset$, every BFS $x^*$ of the polytope above satisfies at least one of the following two conditions:
\end{lemma}
\begin{enumerate}[leftmargin=*,labelindent=.4em]
    \item[(i)]  $x^*_{\alpha g} \in \{0,1\}$ for some $(\alpha, g)\in E$, or
    \item[(ii)] $\sum_{\alpha: (\alpha, g)\in E} x^*_{\alpha g} \leq |g|$ for some nonempty group $g \in \mathcal{G}'$.
\end{enumerate}
\begin{proof}
    Fix a BFS $x^*$, and assume that Condition~(i) does not hold.
    Since $|E|$ constraints must be tight, i.e., hold with equality, at a BFS, the number of agent and good constraints must be at least $|E|$:  $\sum_{g \in \mathcal{G}'} |g| + |M'| \geq |E|$.

Furthermore, since all $x^*_{\alpha g}$ are fractional and each good $\alpha \in M'$ has a total incident weight of $1$ by the good constraints, $\alpha$ must be incident to at least two edges. Hence, $\sum_{g \in \mathcal{G}'} |g| + |M'| \geq |E| \geq 2|M'|$, i.e., $\sum_{g \in \mathcal{G}'} |g| \geq |M'|$.
By summing over all good constraints, we obtain that $\sum_{(\alpha, g)\in E} x^*_{\alpha g} = |M'|$, hence $\sum_{g \in \mathcal{G}'} |g| \geq \sum_{(\alpha, g)\in E} x^*_{\alpha g}$.

Suppose, for contradiction, that Condition~(ii) were also violated.
In this case, each group $g \in \mathcal{G}'$ would satisfy $\sum_{\alpha : (\alpha, g) \in E} x^*_{\alpha g} \geq |g|$ and this inequality would be strict for the nonempty groups.\footnote{Some group is nonempty because $\sum_{g \in \mathcal{G}'} |g| \geq |M'| >0$.}
Summing up over all groups, we obtain $\sum_{(\alpha, g) \in E} x^*_{\alpha g} > \sum_{g \in \mathcal{G}'} |g|$, a contradiction.
\end{proof}

In each iteration, our algorithm finds a BFS $x^*$ for the polytope above, and then proceeds as below.
In the first iteration, we specifically select a BFS representing an fPO allocation, say, by solving an LP that maximizes the sum of all agents' utilities over the polytope. Then:
\begin{enumerate}
    \item We delete all edges $(\alpha, g)$ from $E$ for which $x^*_{\alpha g}= 0$.
    \item If $x^*_{\alpha g} = 1$ for some $(\alpha, g) \in E$, we discretely allocate $\alpha$ to $g$, remove $\alpha$ from $M'$ and $(\alpha, g)$ from $E$.
    \item We update $\mathcal{G}'$ by removing the last agent of every group $g$ for which Condition~(ii) from \cref{lemma:prop-cond} holds. 
\end{enumerate}
Since $x^*$ (restricted to the remaining edges) remains feasible for the updated polytope, the polytope remains nonempty, so we can find a new BFS $x^*$ and repeat the process from Step~1 until all goods are allocated.
Since, by \cref{lemma:prop-cond}, each iteration removes either an agent or an edge, the algorithm terminates in polynomially many iterations.

We now state the main theorem and sketch its proof. We defer pseudocode for the algorithm and the formal proof to Appendix \ref{app:iter_alg}.
\begin{restatable}{theorem}{iterround}
\label{thm:iter_round}
    In any group fair division instance with arbitrary group sizes, there exists an \fpo{} allocation which is \prop{i} for every agent $(g, i)$ where $g\in \mathcal{G}$ and $i \leq |g|$. This allocation can be computed in (weakly) polynomial time.
\end{restatable}

\begin{proof}[Proof sketch]
We have already argued that the algorithm makes progress and terminates in polynomially many iterations. Since each iteration involves solving a linear program and some polynomial computation, the total running time is polynomial.
It remains to argue that the resulting allocation satisfies \prop{i} and \fpo{}.

\paragraph{Proportionality.} Fix an agent $(g, i)$.
If the agent gets never eliminated, their agent constraint ensures that the allocation even satisfies \prop{}.
Should the agent get eliminated in some iteration, it must hold that $\sum_{(\alpha, g) \in E} x^*_{\alpha g} \leq i$ for the BFS $x^*$ of this iteration.
Since $x^*$ satisfies the agent's constraint, the bundle $B_g$ already discretely allocated to $g$ before this iteration satisfies
$u_{gi}(B_g) + \sum_{\alpha : (\alpha, g) \in E} x^*_{\alpha g} \, u_{gi}(\alpha) \geq u_{gi}(M)/n$.
Since $\sum_{\alpha : (\alpha, g) \in E} x^*_{\alpha g} \, u_{gi}(\alpha)$ is at most the value of the $i$ most valuable goods outside of $B_g$, the agent is \prop{i}\,---\,even if the final allocation does not give their group any goods in addition to $B_g$.

\paragraph{Fractional Pareto Optimality.}
It is a classic result by \textcite{VARIAN} that a fractional allocation is fPO iff it maximizes a positively weighted sum of agent utilities. (We confirm that this equivalence persists in the group setting.)
Observe that a fractional allocation maximizes the weighted sum of agent utilities with weights $w_{gi}>0$ iff each good $\alpha$ is only allocated among groups $g \in \argmax_{g \in \mathcal{G}} \sum_{i \in [|g|]} w_{gi} \, u_{gi}(\alpha)$.
It follows that, if the fractional allocation $x$ is fPO and if, for another fractional allocation $x'$, $x_{g\alpha} = 0$ implies $x'_{g\alpha} = 0$ for all groups $g$ and goods $\alpha$, then $x'$ is also fPO.
This argument was previously used, for example, by \textcite{AMS20} and \textcite{BFG+22}.

Since the first iteration of the algorithm starts with an fPO $x^*$, and since the algorithm immediately removes all edges that were zero for $x^*$, the final allocation only allocates goods to groups that received a non-zero amount of this good in the initial fractional allocation.
Hence, the allocation found by the algorithm is fPO.
\end{proof}

\subsection{Possibility of \propone{} Allocations}
\label{sec:prop:prop1}

Although we can only prove the existence of a \proptwo{} allocation among couples, we conjecture that \propone{} allocations exist for any number of couples.

Our conjecture is supported, in part, by a failure to find counter-examples by hand and with computer aid.
More importantly, we were able to show the existence of PROP1 allocations for several special cases:

\begin{restatable}{theorem}{propsepcial}
\label{thm:prop1special}
    When each group $g \in \mathcal{G}$ has size $2$, a \propone{} allocation is guaranteed and efficiently computable whenever one of the following conditions holds.
    \begin{itemize}
        \item $m \leq 2n$.
        \item $m$ divides $n$ and $(g, 1)$ and $(g, 2)$ have opposite preference rankings over the goods for all $g \in \mathcal{G}$.
        \item All agents have binary valuations and approve the same number of goods.
        \item All agents have binary valuations and $n = 3$.
    \end{itemize}
\end{restatable}

Since \propone{} is weaker than \efone{}, one might hope for the existence of \propone{} for even larger groups.
However, \propone{} may fail to exist for groups of size three:
\begin{restatable}{theorem}{hardness}
\label{thm:hardness}
For $n \geq 5$ groups of three agents, PROP1 allocations need not exist, even when utilities are binary and the groups are all identical.
Moreover, deciding whether a PROP1 (or EF1) allocation exists is NP-complete for groups of three agents, even for binary utilities.
\end{restatable}
\begin{proof}[Proof sketch.]
We only give the counter-example for five groups of three agents here, and defer the rest to \cref{app:hardness}.
Consider an instance with goods $\{1, \dots, 9\}$ and five groups $g$, each of which has the following valuations:
\begin{center}
        \begin{tabular}{cccccccccc}
        \toprule
         valuation & $1$ & $2$ & $3$ & $4$ & $5$ & $6$ & $7$ & $8$ & $9$ \\
         \midrule
         $u_{g1}$ & $1$ & $1$ & $1$ & $1$ & $1$ & $1$ & $0$ & $0$ & $0$\\
         $u_{g2}$ & $1$ & $1$ & $1$ & $0$ & $0$ & $0$ & $1$ & $1$ & $1$ \\
         $u_{g3}$ & $0$ & $0$ & $0$ & $1$ & $1$ & $1$ & $1$ & $1$ & $1$ \\
         \bottomrule
    \end{tabular}
\end{center}
Each agent has a total valuation of $6$ and her proportional share is $\frac{6}{5} > 1$.
Hence, each agent must receive at least one good with value $1$ to be \propone{}.
As there are $9$ goods and $5$ groups, one group receives only one good.
By construction, this good has zero value for one of the agents in the group, implying a \propone{} allocation cannot be achieved.
\end{proof}
    
\section{Experiments}
\label{sec:empirics}
\begin{figure}
    \centering
    \includegraphics[width=0.7\columnwidth]{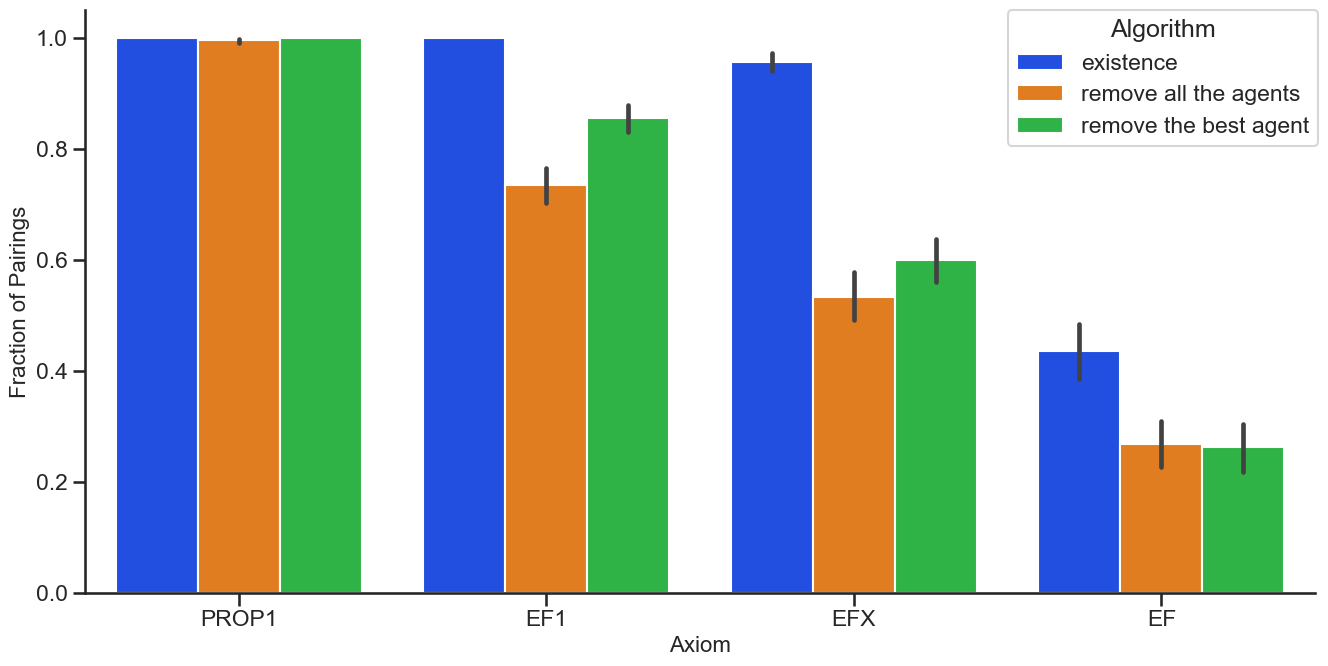}
    \caption{Fraction of pairings for which fair allocations exist or are found by one of two algorithms, averaged over all Spliddit instances. Axioms imply axioms to their left. Error bars indicate 95\% confidence intervals (bootstrapping).}
    \label{fig:plot}
\end{figure}
We now use real-world preference data to empirically examine how often fair allocations exist for practical allocation problems among couples and if our iterative-rounding algorithm exceeds its theoretical guarantees in practice.
Our dataset consists of all allocation problems for indivisible goods (over individuals) submitted to the website \emph{Spliddit}~\cite{GP14a,Shah17} as of June 2025.
To allow us to meaningfully group agents, we consider only Spliddit instances with at least four agents.
The remaining data consists of 254 instances, whose number of agents ranges between 4 and 15 (median: 5) and whose number of goods ranges between 1 and 59 (median: 6).
See \cref{app:experiments} for more details on our data and experiments.

We transform each Spliddit instance into fair allocation problems over couples by iterating over all partitions of agents into pairs (if the number of agents is odd, one agent remains on their own), considering 1000 random pairings if the number of pairings exceeds this number.
Since the different pairings of the same Spliddit instance produce correlated observations, we do not treat them as independent datapoints.
Instead, we calculate for each Spliddit instance the fraction of its pairings that satisfies some property (say, EF1), and report averages over these fractions of pairings.
In \cref{fig:plot}, we display the average fractions of pairs for the existence of several fairness axioms, and for whether these axioms are achieved by two variants of our iterative rounding algorithm.
In \cref{app:experiments}, we show that the patterns remain similar when restricting to instances with many or few agents, or with many or few goods.

While our \cref{lem:more couples} shows that \efone{} allocations do not exist for all fair allocation instances among couples, such allocations seem to exist for most practical problems.
Strikingly, we find \efone{} allocations (hence also \propone{} allocations) for each of the over 13,000 instance--pairing combinations we study.
We also tested the frequency of allocations satisfying EF and EFX,\footnote{An allocation $\{B_{g'}\}_{g'\in \mathcal{G}}$ is EFX for $(g,i)$ if removing any good $\alpha \in B_{g'}$ with $u_{gi}(\alpha) > 0$, from $B_{g'}$, eliminates the envy of the agent.} an axiom between \ef{} and \efone{}, whose existence is a tantalizing open question~\cite{CKM+19,CGM20} in the individual setting. 
As shown by the blue bars in \cref{fig:plot}, EFX exists for 96\% of pairings on average, whereas EF is rarer at 44\%.
    
In \cref{sec:prop}, we proposed a natural iterative-rounding algorithm with proportionality and efficiency guarantees.
To test this algorithm's usefulness in practice, we apply it to the same datasets, and report the fraction of pairings for which the algorithm satisfies each fairness axiom.
The orange bars in \cref{fig:plot} (``remove all the agents'') represent a direct implementation of our algorithm.
Though the algorithm only guarantees \prop{2} for the second agents in the worst case, it satisfies \prop{1} almost always on our data (99\% of pairings on average).
The algorithm even finds \efone{} (73\% of pairings) and EFX (53\%) reasonably often, though substantially less often than the existence of these axioms, which is to be expected since the algorithm does not avoid envy.

We also repeated the experiment with a variant of the iterative-rounding algorithm, in which we do not immediately eliminate the last agent from all groups satisfying condition~(ii) of \cref{lemma:prop-cond}.
Instead, we find the group with the lowest incident weight among eligible groups, and eliminate only a single agent, namely the remaining agent in this group with the largest utility from already discretely allocated goods.
Heuristically, this might lead to fairer allocations by deferring when we drop the constraints of agents who have not yet reached proportionality.

The performance of this variant is shown in \cref{fig:plot} by the green bar (``remove the best agent'').
Eliminating only one agent per iteration leads to \propone{} allocations on all our considered instances and pairings.
For EF1, the variant increases the average fraction of pairings from 73\% to 86\%, an increase clearly beyond the confidence intervals of both estimates (see figure).
The change also moderately increases the fraction of EFX pairings from 53\% to 60\%, and has no discernible effect on EF.
In light of these improvements, it would be interesting to study in future work if starting from \fpo{} allocations other than the one with maximum utilitarian welfare and using other heuristics for eliminating agents can lead to even better practical performance.

\section{Conclusion}
\label{sec:conclusion}
We studied the allocation of indivisible goods among small groups, a well-motivated setting that most prior works, due to their focus on asymptotic bounds in the group size, have left largely unexplored.
For two couples and envy freeness, or for any number of small groups and proportionality, we showed that fairness axioms must not be relaxed by much more than in the individual setting to guarantee existence.

Though our hope of \efone{} existence for arbitrary numbers of couples did not materialize, our work leaves open many possibilities for positive results.
For example, we do not know if \efone{} allocations exist for all allocation problems over couples with binary valuations, whether \propone{} allocations exist for any number of couples with additive utilities (as we believe), or whether, say, it is possible to guarantee \efone{} for one partner and \propone{} for the other in each couple.

\printbibliography

\section*{Appendix}
\appendix
\section{Proof of \cref{thm:2v2}: Existence of \efone{} for Two Couples} \label{thm:2v2proof}
The following lemma is a simplified version of Lemma~4.1 in \textcite{KSV20}.
\begin{lemma}\label{lem:f1ef1}
   Every allocation that contains exactly one good from each pair $\{2j-1, 2j\}$ for all $j \in [m/2]$ is \efone{} for agent $(f, 1)$.
    \begin{proof}
        Let $\alpha_1, \alpha_2, \dots, \alpha_{m/2}$ and $\beta_1, \beta_2, \dots, \beta_{m/2}$ be the goods $f$ and $s$ receives respectively, where $\{\alpha_j, \beta_j\} = \{2j - 1, 2j\}$. As $u_{f1}(\alpha_j) \geq u_{f1}(\beta_{j+1})$, we have $\sum_{j =1}^{m/2} u_{f1}(\alpha_j) 
            \geq \sum_{j = 1}^{m/2 - 1} u_{f1}(\beta_{j+1})$.
        Therefore, once $\beta_1$, the most valuable good of $s$, is taken from them, agent $(f, 1)$ is no longer envious.
    \end{proof}
\end{lemma}
For convenience, we repeat our LP formulation here:
{
\footnotesize
\begin{align}\label{LPd}
    \text{max}~&d \notag\\
     \text{s.t.}~&\sum_{\mathclap{j\in [m/2]}} (2y_j\!-\!1) \, u_{f2}(2j\!-\!1) + (1\!-\!2y_j) \, u_{f2}(2j)\geq \mathrlap{d}\notag\\
        &\sum_{\mathclap{j\in [m/2]}} (1\!-\!2y_j) \, u_{si}(2j\!-\!1) + (2y_j\!-\!1) \, u_{si}(2j) \geq d & i=1, 2\notag\\
        &0 \leq y_j \leq 1 & \mathllap{j = 1, \dots, m/2.} \tag{LP$d$}
\end{align}
}
\twovstwo*
\begin{proof}
    Let $(y^*, d^*)$ be an optimal BFS for the LP, which can be found efficiently (\cref{prop:optbfs}). 
Since setting all $y_j = 1/2$ and $d = 0$ is feasible, we know that $d^*$ is nonnegative and that $y^*$ describes a fractional allocation that is \ef{} for the three agents. We will round an optimal BFS of \ref{LPd} to an \efone{} allocation in which each group receives exactly one good from $\{2j-1, 2j\}$ for every $j\in [m/2]$. The resulting allocation is balanced and, by \cref{lem:f1ef1}, ensures \efone{} for agent $(f, 1)$.
Therefore, we focus on proving \efone{} for the other three agents.
   
Since $y^*$ has at most two fractional values, we can fix $1 \leq \alpha, \beta \leq m/2$ such that all variables except for $y^*_\alpha, y^*_\beta$ are integral. Let $I_f$ and $I_s$ be the set of remaining goods that are allocated entirely to the first and second group, respectively.
We can assume w.l.o.g.\ that $y^*_\alpha, y^*_\beta \geq 1/2$.\footnote{Otherwise, one can swap the roles of, say, goods $2 \alpha - 1$ and $2 \alpha$, which keeps $(f,1)$ \efone{} by \cref{lem:f1ef1}.}
For convenience, set $\alpha_f \coloneqq 2\alpha-1$ to be the good of which group $f$ receives a $y^*_\alpha$ fraction and group $s$ receives a $1 - y^*_\alpha$ fraction; $\alpha_s\coloneqq 2\alpha$ to be the good of which $s$ receives a $y^*_\alpha$ fraction and $f$ receives a $1-y^*_\alpha$ fraction; and analogously for $\beta_f, \beta_s$.
This allocation is illustrated in \cref{fig:appLPd}.

We now slightly modify the allocation $y^*$ to define a partial fractional allocation $z^*$, which will make the next steps of the proof more intuitive. Picture each group $f$ and $s$ transferring the overlapping parts of each good to a \emph{bank}. These overlapping portions include a $(1 - y^*_\alpha)$ fraction of $\alpha_f$ and $\alpha_s$, and a $(1 - y^*_\beta)$ fraction of $\beta_f$ and $\beta_s$. Consequently, each group $g$ is left with $z^*_\alpha \coloneqq y^*_\alpha - (1 - y^*_\alpha)$ of $\alpha_g$ and $z^*_\beta \coloneqq y^*_\beta - (1 - y^*_\beta)$ of $\beta_g$.
Since both groups surrender the same amount of each good, the envy of the agents is unaffected. We will then show how $z^*$ can be rounded to an \efone{} allocation. During the rounding process, if a good from $\{\alpha_f, \alpha_s, \beta_f, \beta_s\}$ is assigned to group $g$, then $g$ still receives the entire good \,---\, part from the groups, and part from the bank. Note that this modification of the initial fractional allocation does not affect the logic of the proof and is just a matter of presentation. $z^*$ is illustrated in \cref{fig:appLPdz}.

\begin{figure}[tb]
\begin{subfigure}[t]{0.38\textwidth}
\resizebox{\textwidth}{!}{
        \centering
\begin{tikzpicture}[
    node distance=0.4cm and 0.8cm,
    every node/.style={font=\large}
]

\begin{scope}[shift={(-7,0)}]
    \node at (0, 4.2) {\large Group $f$};
    \node (alpha_f_top) at (-1, 2.35) {$\alpha_f$\vphantom{$\alpha_{fs}$}};
    \node (box_green) at (-1, 3.2) {\tikz \pic{gift={boxcolor=mygreen, bowcolor=mybowcolor, split=0.6, topopacity=0.07}};};
    \draw[decorate, decoration={brace, amplitude=3pt}]  (-1.5, 3.01) -- (-1.5, 3.32);
    \node[left] at (-1.5,3.3) {\footnotesize{$1{-}y^*_{\alpha}$}};
    \node (box_red) at (1.5, 3.2) {\tikz \pic{gift={boxcolor=myred, bowcolor=mybowcolor, split=0.4, topopacity=0.07}};};
    \node (alpha_s_top) at (1.5, 2.35) {$\alpha_s$\vphantom{$\alpha_{fs}$}};
    \draw[decorate, decoration={brace, amplitude=3pt}]  (-1.5, 2.58) -- (-1.5, 3.01);
    \node[left] at (-1.5,2.85) {\footnotesize{$y^*_{\alpha}$}};
    \draw[thick] (box_red) -- (box_green);
    \node (beta_f_bottom) at (-1, 0.95) {$\beta_f$\vphantom{$\beta_{fs}$}};
    \node (box_blue) at (-1, 1.8) {\tikz \pic{gift={boxcolor=myblue, bowcolor=mybowcolor, split=0.75, topopacity=0.07}};};
    \draw[decorate, decoration={brace, amplitude=3pt}]  (-1.5, 1.74) -- (-1.5, 1.92);
    \node[left] at (-1.5,1.9) {\footnotesize{$1{-}y^*_{\beta}$}};
    \node (box_orange) at (1.5, 1.8) {\tikz \pic{gift={boxcolor=myorange, bowcolor=mybowcolor, split=0.25, topopacity=0.07}};};
    \node (beta_s_bottom) at (1.5, 0.95) {$\beta_s$\vphantom{$\beta_{fs}$}};
    \draw[decorate, decoration={brace, amplitude=3pt}]  (-1.5, 1.18) -- (-1.5, 1.74);
    \node[left] at (-1.5,1.45) {\footnotesize{$y^*_{\beta}$}};
    \draw[thick] (box_blue) -- (box_orange);
    \begin{scope}[yshift=0.8cm]
        \draw (-1.58, -3.05) rectangle (-0.41, -0.2);
        \node at (-0.95, -3.4) {$I_f$};
        \node (l1) at (-1, -0.75) {\tikz \pic{gift={boxcolor=darkgray, bowcolor=mybowcolor}};};
        \node (r1) at (1.5, -0.75) {\tikz \pic{gift={boxcolor=darkgray, bowcolor=mybowcolor, bottomopacity=0.07, topopacity=0.07}};};
        \draw[thick] (l1) -- (r1);
        \node (l2) at (-1, -2.25) {\tikz \pic{gift={boxcolor=darkgray, bowcolor=mybowcolor}};};
        \node (r2) at (1.5, -2.25) {\tikz \pic{gift={boxcolor=darkgray, bowcolor=mybowcolor, bottomopacity=0.07, topopacity=0.07}};};
        \draw[thick] (l2) -- (r2);
        \node at (0.25, -1.5) {\LARGE{$\vdots$}};
    \end{scope}
\end{scope}

\begin{scope}[shift={(-2,0)}]
    \node at (0, 4.2) {\large Group $s$};
    \node (alpha_f_top_s) at (-1, 2.35) {$\alpha_f$};
    \node (box_green_s) at (-1, 3.2) {\tikz \pic{gift={boxcolor=mygreen, bowcolor=mybowcolor, split=0.6, bottomopacity=0.07}};};
    \node (box_red_s) at (1.5, 3.2) {\tikz \pic{gift={boxcolor=myred, bowcolor=mybowcolor, split=0.4, bottomopacity=0.07}};};
    \node (alpha_s_top_s) at (1.5, 2.35) {$\alpha_s$};
    \draw[decorate, decoration={brace, amplitude=3pt}]  (2, 3.32) -- (2, 2.89);
    \node[left] at (2.6,3.2) {\footnotesize{$y^*_{\alpha}$}};
    \draw[thick] (box_red_s) -- (box_green_s);
    \draw[decorate, decoration={brace, amplitude=3pt}]  (2, 2.89) -- (2, 2.58);
    \node[left] at (2.97,2.75) {\footnotesize{$1{-}y^*_{\alpha}$}};
    \node (beta_f_bottom_s) at (-1, 0.95) {$\beta_f$};
    \node (box_blue_s) at (-1, 1.8) {\tikz \pic{gift={boxcolor=myblue, bowcolor=mybowcolor, split=0.75, bottomopacity=0.07}};};
    \draw[decorate, decoration={brace, amplitude=3pt}]  (2, 1.92) -- (2, 1.37);
    \node[left] at (2.6,1.7) {\footnotesize{$y^*_{\beta}$}};
    \node (box_orange_s) at (1.5, 1.8) {\tikz \pic{gift={boxcolor=myorange, bowcolor=mybowcolor, split=0.25, bottomopacity=0.07}};};
    \node (beta_s_bottom_s) at (1.5, 0.95) {$\beta_s$};
    \draw[thick] (box_blue_s) -- (box_orange_s);
    \draw[decorate, decoration={brace, amplitude=3pt}]  (2, 1.37) -- (2, 1.18);
    \node[left] at (2.97,1.25) {\footnotesize{$1{-}y^*_{\beta}$}};
    \begin{scope}[yshift=0.8cm]
        \draw (0.92, -3.05) rectangle (2.1, -0.2);
        \node at (1.5, -3.4) {$I_s$};
        \node (l1_s) at (-1, -0.75) {\tikz \pic{gift={boxcolor=darkgray, bowcolor=mybowcolor, bottomopacity=0.07, topopacity=0.07}};};
        \node (r1_s) at (1.5, -0.75) {\tikz \pic{gift={boxcolor=darkgray, bowcolor=mybowcolor}};};
        \draw[thick] (l1_s) -- (r1_s);
        \node (l2_s) at (-1, -2.25) {\tikz \pic{gift={boxcolor=darkgray, bowcolor=mybowcolor, bottomopacity=0.07, topopacity=0.07}};};
        \node (r2_s) at (1.5, -2.25) {\tikz \pic{gift={boxcolor=darkgray, bowcolor=mybowcolor}};};
        \draw[thick] (l2_s) -- (r2_s);
        \node at (0.25, -1.5) {\LARGE{$\vdots$}};
    \end{scope}
\end{scope}
\end{tikzpicture}%
}
\caption{Illustration of $y^*$. Each good is paired with the good next to it. The bright, solid portion of a box represents the fraction of the good a group receives, while the faded portion represents the fraction allocated to the other group.}\label{fig:appLPd}
\end{subfigure}
\hfill
\begin{subfigure}[t]{0.57\textwidth}
    \centering
    \resizebox{\textwidth}{!}{
\begin{tikzpicture}[
    node distance=0.4cm and 0.8cm,
    every node/.style={font=\large}
]

\begin{scope}[shift={(-7,0)}]
    \node at (0, 4.2) {\large Group $f$};
    \node (alpha_f_top) at (-1, 2.35) {$\alpha_f$\vphantom{$\alpha_{fs}$}};
    \node (box_green) at (-1, 3.2) {\tikz \pic{gift={boxcolor=mygreen, bowcolor=mybowcolor, split=0.2, topopacity=0.07}};};

    \node (box_red) at (1.5, 3.2) {\tikz \pic{gift={boxcolor=myred, bowcolor=mybowcolor, topopacity=0.07, bottomopacity=0.07}};};
    \node (alpha_s_top) at (1.5, 2.35) {$\alpha_s$\vphantom{$\alpha_{fs}$}};
    \draw[decorate, decoration={brace, amplitude=3pt}]  (-1.5, 2.58) -- (-1.5, 2.72);
    \node[left] at (-1.5,2.65) {\footnotesize{$z^*_{\alpha}=2y^*_\alpha{-}1$}};
    \draw[thick] (box_red) -- (box_green);
    \node (beta_f_bottom) at (-1, 0.95) {$\beta_f$\vphantom{$\beta_{fs}$}};
    \node (box_blue) at (-1, 1.8) {\tikz \pic{gift={boxcolor=myblue, bowcolor=mybowcolor, split=0.5, topopacity=0.07}};};
    \node (box_orange) at (1.5, 1.8) {\tikz \pic{gift={boxcolor=myorange, bowcolor=mybowcolor, topopacity=0.07, bottomopacity=0.07}};};
    \node (beta_s_bottom) at (1.5, 0.95) {$\beta_s$\vphantom{$\beta_{fs}$}};
    \draw[decorate, decoration={brace, amplitude=3pt}]  (-1.5, 1.18) -- (-1.5, 1.55);
    \node[left] at (-1.5,1.35) {\footnotesize{$z^*_\beta = 2y^*_{\beta}{-}1$}};
    \draw[thick] (box_blue) -- (box_orange);
    \begin{scope}[yshift=0.8cm]
        \draw (-1.58, -3.05) rectangle (-0.41, -0.2);
        \node at (-0.95, -3.4) {$I_f$};
        \node (l1) at (-1, -0.75) {\tikz \pic{gift={boxcolor=darkgray, bowcolor=mybowcolor}};};
        \node (r1) at (1.5, -0.75) {\tikz \pic{gift={boxcolor=darkgray, bowcolor=mybowcolor, bottomopacity=0.07, topopacity=0.07}};};
        \draw[thick] (l1) -- (r1);
        \node (l2) at (-1, -2.25) {\tikz \pic{gift={boxcolor=darkgray, bowcolor=mybowcolor}};};
        \node (r2) at (1.5, -2.25) {\tikz \pic{gift={boxcolor=darkgray, bowcolor=mybowcolor, bottomopacity=0.07, topopacity=0.07}};};
        \draw[thick] (l2) -- (r2);
        \node at (0.25, -1.5) {\LARGE{$\vdots$}};
    \end{scope}
\end{scope}

\begin{scope}[shift={(-2,0)}]
    \node at (0, 4.2) {\large Group $s$};
    \node (alpha_f_top_s) at (-1, 2.35) {$\alpha_f$};
    \node (box_green_s) at (-1, 3.2) {\tikz \pic{gift={boxcolor=mygreen, bowcolor=mybowcolor, bottomopacity=0.07, topopacity=0.07}};};
    \node (box_red_s) at (1.5, 3.2) {\tikz \pic{gift={boxcolor=myred, bowcolor=mybowcolor, split=0.2, topopacity=0.07}};};
    \node (alpha_s_top_s) at (1.5, 2.35) {$\alpha_s$};
    \draw[thick] (box_red_s) -- (box_green_s);
    \draw[decorate, decoration={brace, amplitude=3pt}]  (2, 2.72) -- (2, 2.58);
    \node[left] at (2.58,2.66) {\footnotesize{$z^*_\alpha$}};
    \node (beta_f_bottom_s) at (-1, 0.95) {$\beta_f$};
    \node (box_blue_s) at (-1, 1.8) {\tikz \pic{gift={boxcolor=myblue, bowcolor=mybowcolor, bottomopacity=0.07, topopacity=0.07}};};
    \node (box_orange_s) at (1.5, 1.8) {\tikz \pic{gift={boxcolor=myorange, bowcolor=mybowcolor, split=0.5, topopacity=0.07}};};
    \node (beta_s_bottom_s) at (1.5, 0.95) {$\beta_s$};
    \draw[thick] (box_blue_s) -- (box_orange_s);
    \draw[decorate, decoration={brace, amplitude=3pt}]  (2, 1.56) -- (2, 1.18);
    \node[left] at (2.58,1.36) {\footnotesize{$z^*_\beta$}};
    \begin{scope}[yshift=0.8cm]
        \draw (0.92, -3.05) rectangle (2.1, -0.2);
        \node at (1.5, -3.4) {$I_s$};
        \node (l1_s) at (-1, -0.75) {\tikz \pic{gift={boxcolor=darkgray, bowcolor=mybowcolor, bottomopacity=0.07, topopacity=0.07}};};
        \node (r1_s) at (1.5, -0.75) {\tikz \pic{gift={boxcolor=darkgray, bowcolor=mybowcolor}};};
        \draw[thick] (l1_s) -- (r1_s);
        \node (l2_s) at (-1, -2.25) {\tikz \pic{gift={boxcolor=darkgray, bowcolor=mybowcolor, bottomopacity=0.07, topopacity=0.07}};};
        \node (r2_s) at (1.5, -2.25) {\tikz \pic{gift={boxcolor=darkgray, bowcolor=mybowcolor}};};
        \draw[thick] (l2_s) -- (r2_s);
        \node at (0.25, -1.5) {\LARGE{$\vdots$}};
    \end{scope}
\end{scope}

\begin{scope}[shift={(3,0)}]
    \node at (-1, 4.2) {\large Bank};
    \node (alpha_f_top_s) at (-1, 2.35) {$\alpha_f$};
    \node (box_green_s) at (-1, 3.2) {\tikz \pic{gift={boxcolor=mygreen, bowcolor=mybowcolor, split =0.8, topopacity=0.07}};};
    \draw[decorate, decoration={brace, amplitude=3pt}]  (-0.5, 3.17) -- (-0.5, 2.58);
    \node[left] at (2.1,2.86) {\footnotesize{$1{-}z^*_\alpha=2(1{-}y^*_\alpha)$}};
    \node (beta_f_bottom_s) at (-1, 0.95) {$\alpha_s$};
    \node (box_blue_s) at (-1, 1.8) {\tikz \pic{gift={boxcolor=myred, bowcolor=mybowcolor, split = 0.8, topopacity=0.07}};};
    \draw[decorate, decoration={brace, amplitude=3pt}]  (-0.5, 1.77) -- (-0.5, 1.18);
    \node[left] at (2.1,1.45) {\footnotesize{$1{-}z^*_\alpha=2(1{-}y^*_\alpha)$}};
    \begin{scope}[yshift=0.8cm]
        \node (l1_s) at (-1, -0.75) {\tikz 
        \pic{gift={boxcolor=myblue, bowcolor=mybowcolor, split= 0.5, topopacity=0.07}};};
        \draw[decorate, decoration={brace, amplitude=3pt}]  (-0.5, -1) -- (-0.5, -1.37);
        \node[left] at (2.1,-1.19) {\footnotesize{$1{-}z^*_\beta=2(1{-}y^*_\beta)$}};
        \node (alpha_s_top_s) at (-1, -1.6) {$\beta_f$};
        \node (l2_s) at (-1, -2.25) {\tikz \pic{gift={boxcolor=myorange, bowcolor=mybowcolor, split=0.5, topopacity=0.07}};};
        \draw[decorate, decoration={brace, amplitude=3pt}]  (-0.5, -2.5) -- (-0.5, -2.87);
        \node[left] at (2.1, -2.69) {\footnotesize{$1{-}z^*_\beta=2(1{-}y^*_\beta)$}};
        \node (alpha_s_top_s) at (-1, -3.1) {$\beta_s$};
    \end{scope}
\end{scope}
\end{tikzpicture}%
}
\caption{Illustration of $z^*$, where the bright, solid region shows the part of the good held by the group or the bank.}\label{fig:appLPdz}
\end{subfigure}
\caption{}
\end{figure}

We use $g$ to denote the group under consideration, and $g'$ for the other group. When we say a bundle $B$ is \ef{} (or \efone{}) for $(g, i)$, we mean the allocation $\{B, M\setminus B\}$ is \ef{} (or \efone{}) for that agent.
We are now ready to describe the rounding procedure. As we intended to give each groups exactly one good from each pair of goods, we have four rounding options, where $f$ receives $\{\alpha_f, \beta_f\}, \{\alpha_f, \beta_s\}, \{\alpha_s, \beta_f\}$, and $\{\alpha_s, \beta_s\}$, and $s$ receives the compliment. Among these, we refer to the one that assigns $\{\alpha_g,\beta_g\}$ to group $g$ \emph{natural rounding}. We have two main cases:
    \paragraph*{Case 1: $z^*_{\alpha} + z^*_{\beta} \geq 1$.} In this case, each group $g$ receives at least one combined unit of good from $\{\alpha_g, \beta_g\}$. As a result, applying the natural rounding adds at most one unit of combined goods to each group’s bundle, which increases the envy of other agents by no more than one unit. For the agent $(g, i)$ we have:
    {
    \footnotesize
    \begin{align*}
            u_{gi}(I_g) + u_{gi}(\alpha_g) + u_{gi}(\beta_g)
            \geq 
            &\ u_{gi}(I_g) + z^*_{\alpha}\, u_{gi}(\alpha_g) + z^*_{\beta}\, u_{gi}(\beta_g) \tag{as $z^*_{\alpha}, z^*_{\beta} \in [0, 1]$}\\
            \geq &\ u_{gi}(I_{g'}) + z^*_{\alpha}\, u_{gi}(\alpha_{g'}) + z^*_{\beta}\, u_{gi}(\beta_{g'}) \tag{by envy-freeness of $z^*$}\\
            \geq &\ u_{gi}(I_{g'}) + (z^*_{\alpha} + z^*_{\beta})\, \min(u_{gi}(\alpha_{g'}), u_{gi}(\beta_{g'})) \\
            \geq &\ u_{gi}(I_{g'}) + \min(u_{gi}(\alpha_{g'}), u_{gi}(\beta_{g'})) \tag{as $z^*_{\alpha} + z^*_{\beta} \geq 1$}
    \end{align*}
        }
        
        \paragraph*{Case 2: $z^*_{\alpha} + z^*_{\beta} < 1$.}
        We call agent $(g, i)$ \emph{unhappy} with $\alpha_g$ if they prefer $\alpha_{g'}$ over it (where $g'$ is the other group) and unhappy with $\beta_g$ if they prefer $\beta_{g'}$. We have the following observations regarding our three agents:
        \begin{enumerate}[leftmargin=*,labelindent=.65em]
        \item[(A)] If an agent $(g,i)$ is unhappy with both $\alpha_g$ and $\beta_g$, they are \efone{} for all rounding options except the natural one.
            \begin{proof}
                As the agents is unhappy with both $\alpha_g$ and $\beta_g$, we have $u_{gi}(\alpha_g) \leq u_{gi}(\alpha_{g'})$, and $u_{gi}(\beta_{g}) \leq u_{gi}(\beta_{g'})$.  Also, by envy-freeness of $(g, i)$ we have
                {
                \footnotesize
                \begin{align}
                u_{gi}(I_g) + z^*_{\alpha}\, u_{gi}(\alpha_g) + z^*_{\beta}\, u_{gi}(\beta_g)
                \geq u_{gi}(I_{g'}) + z^*_{\alpha}\, u_{gi}(\alpha_{g'}) + z^*_{\beta}\, u_{gi}(\beta_{g'})\label{eq:ef}
                \end{align}
                }
                
                which implies $u_{gi}(I_g) \geq u_{gi}(I_{g'})$. 
                 Hence, $(g,i)$ is \efone{} as long as she receives one of her preferred goods $\alpha_{g'}$ or $\beta_{g'}$.
            \end{proof}
            \item[(B)] If an agent is happy with at least one of $\alpha_g$ or $\beta_g$, they are \efone{} for the natural rounding and at least one other rounding option.
            \begin{proof}
            Note that as $(g, i)$ is happy with at least one of $\alpha_g$ and $\beta_g$, we have
            {
            \footnotesize
            \begin{align}
                \max(u_{gi}(\alpha_g) - u_{gi}(\alpha_{g'}), u_{gi}(\beta_g) - u_{gi}(\beta_{g'})) \geq 0\label{eq:max}
            \end{align}
            }
            
            Hence, by rearranging the terms of \cref{eq:ef} we get
                {
                \footnotesize
                \begin{align*}
                    u_{gi}(I_{g'}) - u_{gi}(I_g)
                    \leq &\ z^*_{\alpha}\, (u_{gi}(\alpha_g) - u_{gi}(\alpha_{g'})) + z^*_{\beta}\, (u_{gi}(\beta_g) - u_{gi}(\beta_{g'}))\\
                    \leq &\ (z^*_{\alpha} +z^*_{\beta})\, \max(u_{gi}(\alpha_g) - u_{gi}(\alpha_{g'}), u_{gi}(\beta_g) - u_{gi}(\beta_{g'}))\\
                    \leq &\ \max(u_{gi}(\alpha_g) - u_{gi}(\alpha_{g'}), u_{gi}(\beta_g) - u_{gi}(\beta_{g'}))\tag{as $z^*_{\alpha} + z^*_{\beta} < 1$ and by \cref{eq:max}}\\
                    =&\  u_{gi}(\alpha_g) - u_{gi}(\alpha_{g'})\tag{w.l.o.g.}
                \end{align*}
                }
                Again, by rearranging the terms, we get $u_{gi}(I_g) + u_{gi}(\alpha_g)
                    \geq u_{gi}(I_{g'}) + u_{gi}(\alpha_{g'})$ which shows both $I_g \cup \{\alpha_g, \beta_{g}\}$, the allocation given by the trivial rounding, and $I_g \cup \{\alpha_g, \beta_{g'}\}$ are \efone{} for $(g,i)$.
            \end{proof}
             \item[(C)] If an agent is not \efone{} under some rounding option, then they become \efone{} under the other three options if we \emph{swap the integral parts} $I_f$ and $I_s$.
            \begin{proof}
            Assume $I_g \cup \{\alpha, \beta\}$ is not \efone{} for $(g, i)$ where $\alpha \in \{\alpha_f, \alpha_s\}$ and  $\beta \in \{\beta_f, \beta_s\}$. And let $\alpha' \in \{\alpha_f, \alpha_s\} \setminus \{\alpha\}$ and $\beta' \in \{\beta_f, \beta_s\} \setminus \{\beta\}$ be the other two goods. As $(g,i)$ is not \efone, we have:
             {
                \footnotesize
                \begin{align*}
                    &u_{gi}(I_g) + u_{gi}(\alpha) + u_{gi}(\beta) < u_{gi}(I_{g'}) + u_{gi}(\alpha') \\
                    &u_{gi}(I_g) + u_{gi}(\alpha) + u_{gi}(\beta) < u_{gi}(I_{g'}) +u_{gi}(\beta')
                \end{align*}
                }
                This implies if $(g, i)$ receives $I_{g'}$ and at least one of the $\alpha'$ or $\beta'$, she will be \efone.
            \end{proof}
        \end{enumerate}
        Given these observations, the analysis splits into two subcases:
        \begin{itemize}
           \item \textbf{One of the three agents is \efone{} under all four rounding options.}  
           As a result of observations (A) and (B), any two agents have at least one \efone{} rounding option in common: either natural rounding is \efone{} for both, or an overlap is assured because one agent is \efone{} under three rounding options, and the other one is \efone{} under at least two rounding options.

            The existence of an \efone{} allocation is guaranteed as one of the agents is \efone{} under all four rounding options, and the other two agents have an \efone{} rounding option in common.
            
            \item \textbf{None of the agents is \efone{} under all four rounding options.}  
            Then by observation (C), swapping $I_f$ and $I_s$ ensures that each agent becomes \efone{} under at least three of the four options—eliminating only one. Hence, there exists at least one rounding option that is \efone{} for all three agents.
\end{itemize}
        Note that, in all rounding schemes above, one good from each pair \( \{2j-1, 2j\} \) is allocated to each group, as intended.
        
        The algorithm runs in polynomial time by first computing an optimal BFS of \ref{LPd}, which can be done in polynomial time by \cref{prop:optbfs}, and then applying a linear-time rounding procedure.
\end{proof}

\section{Proof of \cref{lem:more couples}: Non-Existence of \efone{} for Three or More Couples} \label{app:cex ef1}
\efonecex*
\begin{proof}
        Consider an instance with four special goods $\{1, 2, 3, 4\}$ and $n-2$ goods with value $1$ for every agent. For each group $g$, the utilities $u_{g1}(\cdot)$ and $u_{g2}(\cdot)$ restricted to the special goods follow one of the valuation pairs listed in the table below. Also, each of the below pairs occurs for at least one group.
        \begin{table}[H] 
        \centering
        \footnotesize
        \begin{tabular}{ccccc}
            \toprule
             valuation & $1$ & $2$ & $3$ & $4$\\
             \midrule
             $u_{g1}$ & $2$ & $2$ & $0$ & $0$\\
             $u_{g2}$ & $0$ & $0$ & $2$ & $2$\\
             \midrule
             $u_{g1}$ & $0$ & $2$ & $0$ & $2$\\
             $u_{g2}$ & $2$ & $0$ & $2$ & $0$\\
             \midrule
             $u_{g1}$ & $2$ & $0$ & $0$ & $2$\\
             $u_{g2}$ & $0$ & $2$ & $2$ & $0$\\
             \bottomrule
        \end{tabular}
        \label{tab:cex ef1}
        \end{table}

        Each agent has positive valuation for $n$ goods: two with value $2$ and $n-2$ with value $1$. 
        The agent must receive at least one such good since, otherwise, some other group receives two or more of those goods, violating \efone.

        For the sake of contradiction, suppose that an \efone{} allocation exists. With $n+2$ goods and $n$ groups, and since each group must get at least one good, at most two groups can receive more than one good. Therefore, $n-2$ groups receive only one good, and since this good must have positive value for both agents in the group, it can not be among the special goods $\{1, 2, 3, 4\}$. W.l.o.g.\ let the first $n-2$ groups each get one non-special good. Since no special good is liked by both agents in a group, the remaining two groups, $n-1$ and $n$, must receive two of the special goods each. Now, consider the two goods given to the group $n$. By construction, there must be some agent $(g, i)$ for whom each of these two goods has value $2$. Observe that these two goods must then have utility $0$ for $(g, i)$'s partner $(g, i')$. If $g$ is the group $n$, then $(g, i')$ receives value $0$ and must be envious. Otherwise, if $g$ is another group, $(g, i)$ envies the group $n$ by more than one good.
        \end{proof}

\section{Details of \cref{sec:prop:iterative} and the Iterative Rounding Algorithm}\label{app:iter_alg}
\subsection{On Fractional Pareto optimality}
For a set $S \subseteq \mathbb{R}^n$, we say $x \in S$ is \emph{Pareto optimal} if there exists no $y \in S$ such that $y \neq x$ and $y \geq x$. For completeness, we begin by providing a proof of the following proposition.

\begin{proposition}\label{prop:Pareto}
     Let $P=\{ x\in \mathbb{R}^n: Ax \leq b\}$ be a polyhedron. Then $p \in P$ is \emph{Pareto optimal} iff there exists $w > 0$ such that $p$ maximizes  $w^T x$ over $P$, i.e., $\max\{w^Tx: x\in P\} = w^Tp$. 
\end{proposition}
\begin{proof}
    $\Longleftarrow$: If $w > 0$ and $w^T p = \max\{w^T x : x \in P\}$, then $p$ is Pareto optimal; otherwise, a dominating $p'$ would satisfy $w^T p' > w^T p$.

    $\Longrightarrow$: The normal cone of $P$ at point $p$ is defined as $\mathcal{N}_P(p) = \{w \in \mathbb{R}^n : w^T p \geq w^T x \ \forall x \in P\}$. It consists of all $w \in \mathbb{R}^n$ for which $p$ is a maximizer of $w^T x$ over $P$.
    Let $A^= p = b^=$ be the tight constraints of $Ax \leq b$ at point $p$. It is known that $\mathcal{N}_P(p) = \text{cone}\{a_1^=, a_2^=, \dots, a_k^=\}$ where $a_i^=$ is the $i$-th row of $A^=$. Also, as $p$ is Pareto optimal, at least one of the constraints $Ap \leq b$ must be tight\footnote{Otherwise, $p+\epsilon\mathbf{1}$ will be feasible for small enough $\epsilon$.} implying $\mathcal{N}_P(p)\ne \emptyset$.

    We want to show $\mathcal{N}_P(p)$ contains a $w > 0$. For the sake of contradiction, assume $\mathcal{N}_P(p) \cap \mathbb{R}_{>0}^n = \emptyset$. As both $\mathcal{N}_P(p)$ and $\mathbb{R}_{>0}^n$ are convex and non-empty, by the seperating hyperplane theorem, there exists $b\in \mathbb{R}$ and $c \in \mathbb{R}^n\setminus \{0\}$ such that $c^Tx \leq b \leq c^Ty$ for every $x\in \mathcal{N}_P(p)$ and  $y\in \mathbb{R}^n_{> 0}$.
    As $\epsilon \mathbf{1} \in \mathbb{R}^n_{> 0}$ for every $\epsilon >0$, we have $b \leq \lim_{\epsilon \to 0} c^T \epsilon \mathbf{1} = 0$. Similarly, for an $x\in \mathcal{N}_P(p)$ and $\epsilon > 0$ we have $\epsilon x\in \mathcal{N}_P(p)$ implying $b \geq \lim_{\epsilon \to 0} c^T\epsilon x = 0$ giving $b = 0$. As $c^Ty \geq 0$ for every $y \in \mathbb{R}^n_{> 0}$, then $c \geq 0$. As $c^Tx \leq 0$ for every $x\in \mathcal{N}_P(p)$ and in particular for $a_i^=$, we have that $A^=c \leq 0$ which implies there exists a small enough $\epsilon > 0$ such that $A(p+\epsilon a) \leq b$. Therefore, $p+\epsilon c \in P$, and as $c \geq 0$ and $c\ne 0$, this contradicts the Pareto optimality of $p$. 
\end{proof}
We next prove the following result by \textcite{VARIAN} for our group setting.
\begin{lemma}\label{lem:wuw}
    Considering a fair division instance $(\mathcal{G}, M, \mathcal{U})$, a fractional allocation $x$ is \fpo{} iff there exists a weight vector $(w_{gi})_{g\in \mathcal{G}, i\in [|g|]} > 0$ such that
    {
        \footnotesize
        \begin{align*}
            x_{\alpha g} > 0 \Longrightarrow g\in \arg\max_{g' \in \mathcal{G}} \sum_{i \in [|g'|]} w_{g'i}(\alpha)\, u_{g'i}(\alpha)
        \end{align*}
    }
\end{lemma}
\begin{proof}
    Let $P$ denote the set of all fractional allocations for the instance $(\mathcal{G}, M, \mathcal{U})$, and let $U(\cdot)$ be the linear transformation that maps each fractional allocation $x'$ to its utility profile, i.e., $U(x') = \left(\sum_{\alpha \in M} x'_{\alpha g}\, u_{gi}(\alpha)\right)_{g\in \mathcal{G}, i \in [|g|]}$

    Let $Q = \{U(x'):x'\in P\}$ be the image of $P$ under $U(\cdot)$. Observe that $P$ is a polytope, and since $Q$ is a linear image of $P$, it is also a polytope. A point $x \in P$ is an \fpo{} allocation iff $U(x)$ is a Pareto optimal point of $Q$, which, by \cref{prop:Pareto}, is equivalent to the existence of a weight vector $(w_{gi})_{g\in \mathcal{G}, i\in [|g|]} > 0$ such that $U(x)$ maximizes $w^T U(x')$ over all fractional allocations $x'$. Also, 
    {
            \footnotesize
    \begin{align*}
        w^TU(x)= \sum_{\alpha \in M} \sum_{g\in \mathcal{G}} x_{\alpha g}\, \left(\sum_{i \in [|g|]} w_{gi}(\alpha)\, u_{gi}(\alpha)\right),
    \end{align*}
    }
    
is in fact  a ``weighted utilitarian welfare'' of the fractional allocation $x$. Hence, $w^T U(x')$ is maximized at $U(x)$, iff $x$ allocates each good $\alpha$ among the groups that maximize $\sum_{i \in [|g|]} w_{gi}(\alpha)\, u_{gi}(\alpha)$. That is
    {
            \footnotesize
    \begin{align*}
        x_{\alpha g} > 0 \Longrightarrow g\in \arg\max_{g' \in \mathcal{G}} \sum_{i \in [|g'|]} w_{g'i}(\alpha)\, u_{g'i}(\alpha)
    \end{align*}
    }
\end{proof}
\begin{restatable}{corollary}{lemfpo}
\label{lemma:fpo}
If a fractional allocation $x$ is \fpo{} for an instance $(\mathcal{G}, M, \mathcal{U})$ and $x'$ is another fractional allocation such that $x_{\alpha g} = 0$ implies $x'_{\alpha g} =0$, then $x'$ is also \fpo{}.
\end{restatable}
\begin{proof}
By the contrapositive of the assumption, we have $x'_{\alpha g} > 0 \Longrightarrow x_{\alpha g} > 0$. We can then apply \cref{lem:wuw} to $x'$ using the same weight vector as $x$.
\end{proof}

\subsection{Proof of \cref{thm:iter_round}}\label{app:iter_thm}
We repeat the polytope formulation, along with a pseudocode of the algorithm for convenience.
{
\footnotesize
\begin{align*}\label{PROP-LP-iter}
        &\sum_{\mathclap{\alpha:  (\alpha, g) \in E}} x_{\alpha g} \, u_{gi}(\alpha) \geq \frac{u_{gi}(M)}{n}-u_{gi}(B_g) &\forall g\in \mathcal{G}', i \in [|g|]\notag \\
        &\sum_{\mathclap{g:(\alpha, g) \in E}} x_{\alpha g} = 1 & \forall \alpha \in M'\notag\\
        &0 \leq x_{\alpha g} \leq 1 & \forall (\alpha, g) \in E. \tag{$\mathcal{P}$}
\end{align*}
}

\begin{algorithm}[t]
\caption{Finding an Almost \prop{} Allocation}
\label{alg:algorithm-prop}
\textbf{Input}: A group fair division instance consisting of $\mathcal{G}$, $M$, and a set of valuation functions $\{u_{gi}\}_{g\in \mathcal{G}, i\in |g|}$\\
\textbf{Output}: An \fpo{} and almost \prop{} allocation
\begin{algorithmic}[1] 
\STATE Let $M' = M$, $\mathcal{G}' = \mathcal{G}$, $E = M \times \mathcal{G}$.
\STATE Let $B_g = \emptyset$ for every $g \in \mathcal{G}$ be the current bundle of $g$.
\STATE Let $x^*$ be a BFS of \ref{PROP-LP-iter} corresponding to an \fpo{} fractional allocation.
\WHILE{True}
\FOR{$(\alpha, g) \in E$}
\IF {$x^*_{\alpha g} = 0$}
\STATE $E = E \setminus \{(\alpha, g)\}$
\ELSIF {$x^*_{\alpha g} = 1$}
\STATE $B_g = B_g \cup \{\alpha\}$
\STATE$M' = M' \setminus \{\alpha\}$
\STATE$E = E \setminus \{(\alpha, g)\}$
\ENDIF
\ENDFOR
\STATE Update $\mathcal{G}'$ by removing the last agent of every group $g$ with $|g| > 0$ such that $\sum_{\alpha: (\alpha, g)\in E} x^*_{\alpha g} \leq |g|$.
\IF {$M' = \emptyset$} 
\STATE{\textbf{break}}
\ENDIF
\STATE Let $x^*$ be a BFS of \ref{PROP-LP-iter}.
\ENDWHILE
\STATE \textbf{return}  $\{B_g\}_{g \in \mathcal{G}}$
\end{algorithmic}
\end{algorithm}

\iterround*

We prove the theorem by showing that \cref{alg:algorithm-prop} is well-defined and its output satisfies our desired properties. To be more rigorous, we use $\mathcal{P}_t$, $x^{*t}$, $E_t$, $M'_t$, $\mathcal{G}'_t$, and $B_g^t$ to denote the polytope, the computed BFS, the set of edges, remaining goods, state of the groups, and the bundle already allocated to $g$ at the start of iteration $t$.
\paragraph{Well-Definedness and Running Time} We first show that we can find a BFS of $\mathcal{P}_1$ corresponding to an \fpo{} allocation.
Because of the constraint $0 \leq x_{\alpha g} \leq 1$, $\mathcal{P}_1$ is a polytope and is non-empty, as setting \( x_{\alpha g} = \frac{1}{n} \) for all \( (\alpha, g) \in M \times \mathcal{G} \) is feasible for it. Hence, by \cref{prop:optbfs} we can compute a BFS $x^{*1}$ maximizing the total utilitarian welfare, $\sum_{g\in \mathcal{G}} \sum_{i \in [|g|]} \sum_{\alpha \in M} x_{\alpha g}\, u_{gi}(\alpha)$. This BFS corresponds to an \fpo{} fractional allocation, as if there existed another fractional allocation \( x' \) that Pareto dominates $x^{*1}$, then \( x' \) would also be feasible for $\mathcal{P}_1$ and yield a higher total utility, contradicting optimality of $x^{*1}$.

Next, we prove that as long as there are unallocated goods, i.e., $M' \ne \emptyset$, we can find a BFS $x^*$. Again, by \cref{prop:optbfs}, it suffices to show our polytope \ref{PROP-LP-iter} remains non-empty throughout the algorithm, which we prove by induction. We already showed that $\mathcal{P}_1$ is non-empty. Suppose $\mathcal{P}_t$ is non-empty for some iteration \( t \). We claim $x^{*t}_{E_{t+1}}$, i.e., $x^{*t}$ restricted to the edges of $E_{t+1}$, is feasible for $\mathcal{P}_{t+1}$. Obviously $0 \leq x^{*t}_{E_{t+1}} \leq 1$. Since we only eliminate edges with \( x^{*t}_{\alpha g} \in \{0,1\} \), and we remove \( \alpha \) from the set of goods whenever \( x^{*t}_{\alpha g} = 1 \), the total weight of \( x^{*t}_{E_{t+1}} \) incident to each good in $M'_{t+1}$ is $1$. 
As we allocate $\alpha$ to $g$ only when  $x^{*t}_{\alpha g} = 1$, the agent constraints continue to hold. Rigorously, for every $(g, i)$ we have:
{
\footnotesize
\begin{align*}
    \frac{u_{gi}(M)}{n}-u_{gi}(B^t_g)
    \leq&\sum_{\alpha: (\alpha, g) \in E_t} x^{*t}_{\alpha g}\, u_{gi}(\alpha) \tag{by feasibility of $x^{*t}$ for $\mathcal{P}_t$} \\
    = &\sum_{\alpha: (\alpha, g) \in E_{t+1}} x^{*t}_{\alpha g}\, u_{gi}(\alpha) + u_{gi}(B_g^{t+1} \setminus B_g^t)\tag{as $E_{t+1} = E_t \setminus \{(\alpha, g): x^{*t}_{\alpha g} \in \{0, 1\}\}$ and $B_g^{t+1} = B_g^t \cup \{\alpha: x^{*t}_{\alpha g} = 1\}$}
\end{align*}
}

Hence, $x^{*t}_{E_{t+1}}$ is feasible for $\mathcal{P}_{t+1}$, implying the polytope at the beginning of iteration $t+1$ is non-empty.

By \cref{lemma:prop-cond}, each iteration eliminates either an edge or an agent, implying that the total number of iterations is bounded by \( |M \times \mathcal{G}| + \sum_{g \in \mathcal{G}} |g| \). As in every iteration, we compute a BFS of \ref{PROP-LP-iter}\footnote{The formulation of \ref{PROP-LP-iter} is polynomial in the size of the problem input.} and perform a polynomial-time rounding step, the overall runtime of the algorithm is polynomial.

\paragraph{Proportionality} Consider an agent $(g, i)$ we eliminate at some iteration $t$. Let $B \in M'_t$ be the set of $i$ most valuable remaining goods for $(g, i)$, and if $|M'_t| < i$ let $B = M'_t$. As $(g, i)$ is removed at $t$ the total weight of $x^{*t}$ incident to $g$ is at most $i$. So, 
{
\footnotesize
\begin{align*}
    u_{g_i} (B) \geq \sum_{\alpha: (\alpha, g) \in E_t} x^{*t}_{\alpha g}\, u_{gi}(\alpha) \geq\frac{ u_{gi}(M)}{n}- u_{gi}(B^t_g)
\end{align*}
}

Hence, adding $B$ to $g$'s bundle makes $(g, i)$ proportional. As $B \cap B_g^t = \emptyset$ and $|B| \leq i$, $(g, i)$ is \prop{i} at iteration $t$. Since $g$ only receives more goods in later iterations, $(g, i)$ continues to satisfy \prop{i}.
Now consider an agent $(g, i)$ we never eliminate. In the last iteration $T$, as we allocate all the remaining goods, $x^{*T}$ must be fully integral. Therefore we have
{
\footnotesize
\begin{align*}
    \frac{u_{gi}(M)}{n}- u_{gi}(B^T_g)
    \leq&\ \sum_{\alpha: x^{*T}_{\alpha g} = 1} u_{gi}(\alpha) \tag{by feasibility of $x^{*T}$ for $\mathcal{P}_T$ and $x^{*T} \in \{0,1\}$} \\
    = &\ u_{gi}(B_g^{T+1}\setminus B_g^{T}) \tag{as $B_g^{T+1} = B_g^T \cup \{\alpha: x^{*T}_{\alpha g} = 1\}$}
\end{align*}
}

Hence, $B_{g}^{T+1}$, the final bundle of $g$, is \prop{} for $(g, i)$.

\paragraph{Fractionally Pareto optimality} We show our final allocation is \fpo{} by induction. Define the fractional allocation $x^t$ as the union of $x^{*t}$ and the previously allocated goods, i.e., $x^t_{E_t} = x^{*t}$, $x^t_{\alpha g} = 1$ for every $g\in \mathcal{G}$ and $\alpha \in B_g^t$, and $x^t_{\alpha g} = 0$ on the rest of $(\alpha, g) \in M\times \mathcal{G}$. 
We have $x^1 = x^{*1}$ and we chose $x^{*1}$ to be \fpo{}. Assume $x^t$ is \fpo{} for some $t$. Note that for every $(\alpha, g)$ with $x^t_{\alpha g} = 0$,  $\alpha \notin B_g^{t+1}$, and $(\alpha, g)$ is either already removed or will be removed during iteration $t$. Thus, $x^{t+1}_{\alpha g}$ is also $0$ and by \cref{lemma:fpo} $x^{t+1}$ is \fpo{}.

\section{Special Cases for Existence of \propone{} and \efone{} Among Couples}
We begin by introducing some new definitions and notations. W.l.o.g.\, assume that the number of goods $m$ is divisible by the number of groups $n$; otherwise, we can add dummy goods with zero value for all agents to the set of goods $M$. For each agent $(g, i)$, we define their \emph{Segment Partition} as a partition of the goods into $\frac{m}{n}$ segments, each of size $n$, where the first segment contains their top $n$ most valued goods, the second segment contains the next $n$ most valued goods, and so on. We denote this partition by $\mathcal{S}^{gi} = \{S_1^{gi}, S_2^{g_i}, \dots , S_{m/n}^{gi}\}$ where $S_1^{gi}$ is the set of $n$ most vaued good for $(g, i)$.

\begin{lemma} \label{lem:segprop1}
    If an agent is allocated exactly one good from each segment of their segment partition, they satisfy \propone{}.
    \begin{proof}
         Consider an agent $(g, i)$ and suppose w.l.o.g.\ that $u_{gi}(1) \geq u_{gi}(2) \geq \dots \geq u_{gi}(m)$.
                
        Let $S = \{1, n+1, \dots, m-n+1\}$. We first show that $u_{gi}(S)\geq \frac{u_{gi}(M)}{n}$. To see this, partition the goods into $n$ sets based on their indices modulo $n$. Observe that, the set $S$ is the most valuable of the $n$ sets and thus have a value of at least $\frac{u_{gi}(M)}{n}$.
        
        Now, consider any bundle $B = \{\alpha_1, \alpha_2, \dots, \alpha_{m/n}\}$ such that $\alpha_j \in S^{gi}_j$ is in the $j$-th segment. Let $k$ be the smallest integer such that $k\notin B$. Note that $k$ is either $1$ or we have $\alpha_1 = 1$ and $k = 2$. Thus, $1 \in \{\alpha_1, k\}$ and as both $\alpha_1$ and $k$ are in the first segment, we have $u_{gi}(\{\alpha_1, k\}) \geq u_{gi}(1) + u_{gi}(n+1)$. As $\alpha_j$ belongs to the $j$-th segment and $nj+1$ belongs to the $(j{+}1)$-th segment, we have $u_{gi}(\alpha_j) \geq u_{gi}(nj + 1)$. Hence,
        {
        \footnotesize
        \begin{align*}
            u_{gi}(B \cup \{k\}) &= u_{gi}(\{\alpha_1, k\}) + \sum_{j = 2}^{m/n} u_{gi}(\alpha_j)\\
            &\geq u_{gi}(1) + u_{gi}(n+1) + \sum_{j = 2}^{m/n - 1} u_{gi}(nj + 1)  + u_{gi}(\alpha_{m/n}) \\
            &\geq u_{gi}(S)\\
            &\geq \frac{u_{gi}(M)}{n}
        \end{align*}
        }

        implying $B$ is \propone{} for $(g, i)$.
    \end{proof}
\end{lemma}

\begin{lemma}\label{lem:sameseg}
If for every group $g\in\mathcal{G}$, the segment partition of both $(g,1)$ and $(g,2)$ is the same, i.e., $\mathcal{S}^{g1} = \mathcal{S}^{g2}$, a \propone{} allocation exists.
    \begin{proof}
        Let $\mathcal{S}^g$ be the segment partition of agents in group $g$, and let $\mathcal{S} = \uplus_{g\in \mathcal{G}} \mathcal{S}^g$ be the multiset obtained by the additive union of all $n\cdot \frac{m}{n}$ segments. Construct a bipartite graph $G=(M \cup \mathcal{S}, E)$ where there is an edge between $\alpha\in M$ and $S\in \mathcal{S}$ iff good $\alpha$ belongs to the segment $S$. Note that both parts of $G$ have $m$ nodes and $G$ is $ n$-regular, as each segment contains $n$ goods and each good belongs to exactly one segment in the segment partition of each group. Therefore, $G$ admits a perfect matching $F$.For each edge $(\alpha, S) \in F$, we allocate good $\alpha$ to the group $g$ such that $S \in \mathcal{S}^g$. With this assignment, every agent receives exactly one good from each segment of their segment partition and, by \cref{lem:segprop1}, satisfies \propone{}.
    \end{proof}
\end{lemma}

\begin{lemma}\label{lem:m<=2n}
            If the number of goods is at most the number of agents, i.e., $m \leq 2n$, \propone{} allocation is guaranteed to exist.
            \begin{proof}
            
            If $m = n$, every allocation is \propone{}. Therefore, we suppose $m = 2n$. Using the same reasoning as in \cref{lem:segprop1}, observe that each agent $(g, i)$ will be \propone{} if they receive a good from their $n$ most valued goods, $S^{gi}_1$. We now present an algorithm that finds an allocation where this condition holds. The algorithm works in two phases.

           In the first phase, we greedily assign a good $\alpha$ to a group $g$ such that $\alpha \in S^{g1}_1 \cap S^{g2}_1$, which means $\alpha$ makes both $(g, 1)$ and $(g, 2)$ \propone{}. We then exclude $g$ from further consideration. At the end of the first phase, we either have satisfied every group or every remaining good is among the $n$ most valuable goods of at most one agent in each group. 
           
           In the latter case, we run the second phase. Assume we allocate $k$ goods in the first phase. Note that $k$ is also the number of satisfied groups. Let $M'$ be the set of remaining goods and $N'$ be the set of remaining \emph{agents}. So we have $|M'| = m - k$ and $|N'| = 2n - 2k$.
           
           Consider a bipartite graph $G = (M' \cup N', E)$ where $(\alpha, (g, i)) \in E$ iff $\alpha \in S^{gi}_1$. We next show that Hall's condition holds in $G$ for part $N'$. Note that as $|S^{gi}_1| = n$ and as we removed $k$ goods in the first phase, each remaining agent $(g, i)$ is incident to at least $n-k$ goods. For every subset $A \subseteq N'$, if $|A| \leq n-k$ then $|N_G(A)| \geq n - k \geq |A|$ as the degree of every agent is at least $n-k$, where $N_G(A)$ is the set of neighbors of $A$ in $G$. Otherwise, if $|A| > n-k$, as there are only $n-k$ groups left there exists some group $g$ where both $(g, 1), (g,2) \in A$. Now, note that $(g,1)$ and $(g, 2)$ do not have any common neighbor, as if they did, we would allocate that good to $g$ in the first phase. As each of $(g, 1)$ and $(g, 2)$ are incident to at least $n-k$ goods, we have $|N_G(A)| \geq 2n - 2k = |N'| \geq |A|$. Thus, by Hall's condition, there exists a matching in $G$ that covers $N'$. That is, we can assign each of the remaining agents one of their most $n$ valuable goods, and then we can allocate the remaining goods arbitrarily, achieving \propone{} for every agent.
            \end{proof}
        \end{lemma}
        The next two lemmas assume binary valuation functions. When we say an agent $(g, i)$ likes or approves a good $\alpha$, we mean $u_{gi}(\alpha) = 1$.
        \begin{lemma}\label{lem:needs-k}
            When valuations are binary and all agents approve between $kn+1$ and $(k+1)n$ goods for some given integer $k$, a \propone{} allocation exists.
            \begin{proof}
                Note that each agent must receive at least $k$ goods they like to be \propone{}.
                Similar to the previous lemma, the algorithm consists of two phases. In the first phase, as long as there exists a group $g$ where both $(g, 1)$ and $(g, 2)$ like one of the remaining goods, we allocate that good to $g$. We also make sure no group receives more than $k$ goods.
                At the end of the first phase, we either have satisfied every group or every remaining good is liked by at most one agent in each group. 
                
                In the latter case, we run the second phase. Let $M'$ be the set of remaining goods, let $k_g \leq k$ be the number of goods allocated to $g$ in the first phase, and let $k' = \sum_{g\in \mathcal{G}}k_g$ be the total number of allocated goods. 
                Consider a bipartite graph $G = (M' \cup N, E)$ where $N$ is the set of agents and $(\alpha, (g, i)) \in E$ iff $(g, i)$ likes $\alpha$, i.e., $u_{gi}(\alpha) = 1$.
                Observe that to find a \propone{} allocation, it suffices to find a subset of edges $F$ such that in $G_F = (M \cup N, F)$ each $\alpha\in M'$ has degree at most $1$ and each agent $(g, i)$ has degree at least $k-k_i$. This is equivalent to finding an integral solution of the following polytope:

                {
                \footnotesize
                \begin{align}\label{allkpolytope}
                    &\sum_{\alpha:(\alpha, (g, i)) \in E} x_{\alpha,(g,i)} \geq k - k_g &\forall (g, i)\in N\notag\\
                    &\sum_{(g, i): (\alpha, (g, i)) \in E} x_{\alpha, (g,i)} \leq 1 & \forall \alpha \in M'\notag\\
                    &0 \leq x_{\alpha,(g,i)} \leq 1 & \forall \alpha\in M', (g,i)\in N \tag{$\mathcal{P}$}
                \end{align}
                }

                Notice that the coefficient matrix of \ref{allkpolytope} is the incidence matrix of $G$, which is known to be totally unimodular (TU), and therefore \ref{allkpolytope} is integral, i.e., every BFS of the \ref{allkpolytope} is integral. Thus, it remains to prove \ref{allkpolytope} is non-empty, and then we are done by \cref{prop:optbfs}.
                We claim that \ref{allkpolytope}  contains the following $x$:
                 {
                \footnotesize
                \begin{align*}
                    x_{\alpha, (g, i)} = \begin{cases}
                        \frac{k-k_g}{kn - k'} & \text{$(g, i)$ likes $\alpha$}\\
                        0 & \text{otherwise}
                    \end{cases}
                \end{align*}
                }
                
                Note that for every $g\in \mathcal{G}$ we have $k - k_g \geq 0$ and $kn- k'> 0$. Also
                {
                \footnotesize
                \begin{align*}
                    \frac{k-k_g}{kn -k'} \leq 1 \Longleftrightarrow k-k_g \leq kn -k' \Longleftrightarrow k' - k_g\leq k(n-1)
                \end{align*}
                }

                and the last inequality holds as each group receives at most $k$ goods in the first phase.
                As $(g, i)$ likes at least $kn + 1 - k'$ of the remaining goods, we have
                {
                \footnotesize
                \begin{align*}
                    \sum_{\alpha \in M'} x_{\alpha, (g, i)} \geq (kn+1-k')\cdot\frac{k-k_g}{kn - k'}  \geq k - k_g
                \end{align*}
                }
                
                On the other hand, each good $\alpha \in M'$ is liked by at most one agent of each group $g$, as otherwise we would allocate $\alpha$ to $g$ in the first phase. Hence,
                {
                \footnotesize
                \begin{align*}
                    \sum_{(g, i) \in N} x_{\alpha, (g, i)} \leq \sum_{g \in \mathcal{G}} \frac{k-k_g}{kn - k'} = \frac{kn-\sum_{g\in \mathcal{G}} k_g}{kn - k'} = \frac{kn-m'}{kn-m'} = 1
                \end{align*}
                }
                
                Therefore, the polytope is non-empty, and we can make every agent \propone{}.
            \end{proof}
        \end{lemma}
\begin{lemma}\label{lem:prop1for3}
    When there are $n=3$ couples and their valuation functions are binary, a \propone{} allocation exists.
\end{lemma}
\begin{proof}
    For each group $g \in \mathcal{G}$, we partition the goods into sets of size $n = 3$ as follows. We first form as many sets of three as possible using goods that are valued at $0$ by both $(g, 1)$ and $(g, 2)$. 
    Once fewer than three such goods remain, we continue by forming sets from goods liked only by $(g, 1)$, then from those liked only by $(g, 2)$, and finally from goods liked by both agents. 
    Since at most two goods can remain after each of these four phases, the total number of remaining goods at the end is at most $8$, and each agent $(g,i)$ values at most $4$ of the remaining goods. Given that $m$ is divisible by $n = 3$, the number of remaining goods must be either $0$, $3$, or $6$.
    
    To handle the remaining goods, we distinguish three cases. 
    First, if neither agent values $4$ of the remaining goods, we arbitrarily partition them into sets of three. Second, if exactly one agent $(g,i)$ values $4$ of the goods, we construct one set of three from the goods they like and form another arbitrary set of three from the rest. Finally, if both agents value exactly $4$ goods, the structure must be such that there are two goods liked by both agents, two liked only by $(g, 1)$, and two liked only by $(g, 2)$. In this case, the remaining goods can be split into two sets of three such that each agent fully values one of the sets.
    
    We denote the resulting partition for group $g$ as $\mathcal{S}^g$. Observe that, for an agent $(g,i)$ who values $\ell$ goods, the constructed partition $\mathcal{S}^g$ ensures that in at least $\lceil \ell/3 \rceil - 1$ of the sets, the agent values every good in the set, therefore if $g$ receives exactly one good from each set in  $\mathcal{S}^g$, both of its agents will be \propone{}. To find such an allocation, we construct the same $n$-regular bipartite graph as \cref{lem:sameseg}, and we then find a perfect matching.
\end{proof}

\propsepcial*
\begin{proof}
    The statements are a direct corollary of \cref{lem:m<=2n}, \cref{lem:sameseg}, \cref{lem:needs-k}, and \cref{lem:prop1for3} respectively.
\end{proof}

\subsection*{Special Case for Existence of \efone{}} In the special case where the agents $(g, 1)$ across all groups $g$ have identical valuations, 
\textcite{allocator2023} show that \efone{} allocations do exists, using a variant of envy-cycle elimination due to \textcite{cardinality}. We noticed their algorithm also works for a more general case. A slightly modified version of their algorithm is presented in \cref{alg:doubly-ef1}, adapted to our notation. 
\begin{algorithm}[]
\caption{Finding an \efone{} Allocation Among Couples in a Special Case}
\label{alg:doubly-ef1}
\textbf{Input}: A group fair division instance consisting of $\mathcal{G}$, $M$, and a set of valuation functions $\{u_{gi}\}_{g\in \mathcal{G}, i\in \{1,2\}}$ such that the agents $(g, 1)$ across all group $g$ have identical segment partition $\mathcal{S}^{g1} = \mathcal{S}$ \\
\textbf{Output}:  An \efone{} allocation
\begin{algorithmic}[1]
\STATE Let $G = (V, E)$ be the envy-graph where each vertex represents an agent $(g, 2)$ and $E \leftarrow \emptyset$.
\STATE Initialize the allocation $B_g =\emptyset$ for every $g\in \mathcal{G}$.
\FOR{$S \in \mathcal{S}$}
    \STATE Apply the envy-cycle elimination algorithm to $G$ to obtain a directed acyclic graph.
    \STATE Let $\{i_1, \ldots, i_n\}$ be the second agents $(g, 2)$ in topological order of graph $G$, ensuring each agent does not envy those before them.
    \FOR{$j = 1, 2, \dots, n$}
        \STATE Allocate agent $i_j$, their favorite good $\alpha$ among the remaining goods in $S$.
        \STATE $S = S\setminus \{\alpha\}$
    \ENDFOR
    \STATE Update the envy-graph $G$.
\ENDFOR
\RETURN $\{B_g\}_{g\in \mathcal{G}}$
\end{algorithmic}
\end{algorithm}

The only difference with the original algorithm (\cite[Algorithm 1]{allocator2023}) is that instead of requiring identical valuations among agents $(g, 1)$, we only assume that they have a common segment partition, which is a slightly less restrictive assumption. The \efone{} guarantee for the first agents $(g,1)$ follows directly from \cref{lem:seg-ef1}, which generalizes \cref{lem:f1ef1}. For the second agents $(g, 2)$, \efone{} can be shown using a nearly standard envy-cycle elimination argument, as in \textcite{allocator2023}. 
\begin{lemma} \label{lem:seg-ef1}
    For an agent $(g, i)$, any allocation $\{B_{g'}\}_{{g'}\in \mathcal{G}}$ where each bundle $B_{g'}$ includes exactly one good from every segment $S^{gi}_j$ in the segment partition of $(g, i)$\,---\,$|B_{g'} \cap S^{gi}_j| = 1$ for all $g'\in \mathcal{G}$ and $j \in [m/n]$\,---\,is \ef{1} for the agent.
\end{lemma}
\begin{proof}
    Fix an arbitrary group $g'$. Let us denote the good in the intersection of $B_g$ and $S^{gi}_j$ with $\alpha_j$, and the good in the intersection of $B_{g'}$ and $S^{gi}_j$ with $\alpha_j'$.
    As the $j$-th segment $S^{gi}_j$ is more valuable for $(g,i)$ than the $(j{+}1)$-th segment $S^{gi}_{j+1}$, we have $u_{gi}(\alpha_j) \geq u_{gi}(\alpha_{j+1}')$ for every $j\in [m/n-1]$ which implies $u_{gi}(B_g) = \sum_{j=1}^{m/n} u_{gi}(\alpha_j) \geq \sum_{j=1}^{m/n-1} u_{gi}(\alpha_{j+1}') =u_{gi}(B_{g'}\setminus \{\alpha_1'\})$. Hence, $(g, i)$ does not envy $g'$ by more than one good.
\end{proof}

\section{Proof of \cref{thm:hardness}: Impossibility of \propone{} (or \efone{}) Among Groups of Three}\label{app:hardness}
\hardness*
\begin{proof}
We first present the counter-example. Consider an instance with $n$ group and the set of goods $\{1, \dots, 2n - 1\}$, where for each agent $(g,i)$ we have
 {
                \footnotesize
                \begin{align*}
                    u_{gi}(\alpha) = 0 \iff \alpha \equiv i \pmod{3}
                \end{align*}
                }
                
Implying each agent disapproves either $\lfloor m/3 \rfloor$ or $\lceil m/3 \rceil$ of the goods, and likes at least $\frac{2m}{3} - 1=\frac{4n-5}{3}$ and at most $\frac{2m}{3} +1=\frac{4n + 1}{3} < 2n$ goods. For $n > 5$ we have $\frac{4n-5}{3} > n$, and for $n = 5$, as $m=9$ is divisible by $3$, each agent disapproves exactly $3$ goods, and likes $6 > 5$ goods. In any case, when $n \geq 5$, every agent has a proportional share of larger than $1$ and smaller than $2$. Hence, each agent must receive at least one good with value $1$ to be \propone{}.
As there are $2n-1$ goods and $n$ groups, one group receives only one good. By construction, this good has zero value for one of the agents in the group, implying a \propone{} allocation cannot be achieved.

To prove NP-completeness, we reduce the \emph{$3$-Dimensional Matching (3DM)} problem to the problem of deciding the existence of \propone{} among groups of three with binary valuations. 
3DM is a well-known NP-complete problem, and asks whether, given sets $X_1$, $X_2$, $X_3$ each of size $k$ and a set of triples $T \subseteq X_1 \times X_2 \times X_3$, there exists a matching of size $k$ in which each element of $X_1 \cup X_2 \cup X_3$ appears in exactly one triple. 
Given a 3DM instance $(X_1 \cup X_2 \cup X_3, T)$, we construct a corresponding group fair division instance as follows. The set of goods is $M = X_1 \cup X_2 \cup X_3 \cup Y$, where $|Y| = 3(k + |T| + 2)$. 
The set of groups is $\mathcal{G} = T \cup U$, where $|U| = 2k + |T| + 3$, and each group contains exactly three agents. 
Hence, we have $m=3(k + |T|+2) + 3k$ goods and $n=2(k+|T|+2)-1$ groups in total. 
For every triple $t = (x_1, x_2, x_3) \in T$ and $i\in [3]$, agent $(t, i)$ approves the good $x_i$ and all the goods in $Y$. 
For every $u \in U$ and $i\in [3]$, agent $(u, i)$ approves exactly two-thirds of the goods in $Y$, such that every good in $Y$ is approved by at most two agents from group $u$. To achieve this, we index the goods in $Y$ from $1$ to $3(k + |T| + 2)$ and let agent $(u, i)$ disapprove good $\alpha \in Y$ iff $i \equiv \alpha \pmod{3}$.

We now prove that the 3DM instance $(X_1 \cup X_2 \cup X_3, T)$ has a perfect matching iff the corresponding fair division instance admits a \propone{} allocation. 
Note that in the constructed fair division instance, each agent $(t, i)$ approves more than $n$ but fewer than $2n$ goods, $n < |Y| + 1 < 2n$, while each agent $(u, i)$ approves exactly $\frac{2}{3}|Y| = 2(k + |T| + 2)=n+1$ goods. 
Therefore, to satisfy \propone{}, each agent must receive at least one good they approve. 
Since no good is approved by all three agents in any group $u \in U$, at least two goods must be assigned to each such group.
As a result, $2|U| = 4k + 2|T| + 6$ of the $3(k + |T| + 2)$ goods in $Y$ must be allocated to groups in $U$, which leaves only $|Y| - 2|U| = |T| - k$ goods in $Y$, along with $k$ goods in each set $X_i$. 
With the remaining $|T| - k$ goods from $Y$, we can satisfy at most $|T| - k$ groups in $T$.
The remaining $k$ groups must be satisfied using the $3k$ goods in $X_1 \cup X_2 \cup X_3$, and as each agent $(t,i)$ only approves one good from $X_1 \cup X_2 \cup X_3$, group $t = (x_1, x_2, x_3)$ must receive $\{x_1, x_2, x_3\}$ to satisfy all three agents. 

Therefore, if a \propone{} allocations exists, the $3k$ goods in $X_1 \cup X_2 \cup X_3$, must be divided between $k$ groups of $T$ such that each group $t=(x_1, x_2, x_3)$ receives exactly $\{x_1, x_2, x_3\}$ implying a perfect matching exist in the 3DM instance $(X_1\cup X_2 \cup X_3, T)$.
Conversely, assume a perfect matching exists. For each of the $k$ matched triples $t = (x_1, x_2, x_3)$, assign the goods $\{x_1, x_2, x_3\}$ to group $t$ in the fair division instance. The remaining $|T| - k$ groups in $T$ are assigned the first $|T| - k$ goods in $Y$.
For each group in $U$, assign two of the remaining goods from $Y$ with consecutive indices $j$ and $j+1$. 
This way every agent in $u \in U$ approves at least one of the allocated goods, thus satisfying \propone{}.

We conclude by observing that the same reduction applies to deciding the existence of an \efone{} allocation. If a perfect matching exists, the \propone{} allocation constructed above is also \efone{}. Each agent $(u, i)$ receives one approved good, while any other group $u'\in U$ receives at most two, so $(u, i)$ does not envy $u'$ by more than one good. They also do not envy groups $t \in T$ since $(u,i)$ only approves goods in $Y$ and $t$ receives at most one good from $Y$. An agent $(t, i)$ receives exactly one good they like and envy a group in $U$ by one good, as $U$ receives two approved goods. And, they do not envy another group $t' \in T$, since $t'$ receives either one good of $X_i$ or one good of $Y$. Conversely, if no perfect matching exists, then as shown before, a \propone{} allocation is not possible, which implies an \efone{} allocation does not exist.
\end{proof}

\section{Details on Experiments}
\label{app:experiments}
We received the data through private communication with Spliddit. We used Spliddit data primarily because no publicly available datasets for fair division currently exist.
Our dataset includes all data available on the website as of June 13, 2025. In each instance, every agent's valuation function sums to 1000 across all goods. The only preprocessing performed was discarding the instances containing divisible goods and removal of goods that had zero value for all agents within an instance.
We ran the experiments using a Dell Optiplex with an Intel\textsuperscript{®} Core™ i7-1185G7 CPU, 16 GB of RAM, running Windows 10, version 22H2. Linear programs were optimized with Gurobi 11.0.3. The experimental code is available on GitHub\footnote{\url{https://github.com/HannaYzade/fair-division-among-couples-and-small-groups}}.

\cref{fig:app:plots}
shows the results of our experiment after restricting the number of goods and the groups.
\begin{figure}[H]
    \centering

    \begin{subfigure}[b]{0.47\textwidth}
        \includegraphics[width=\textwidth]{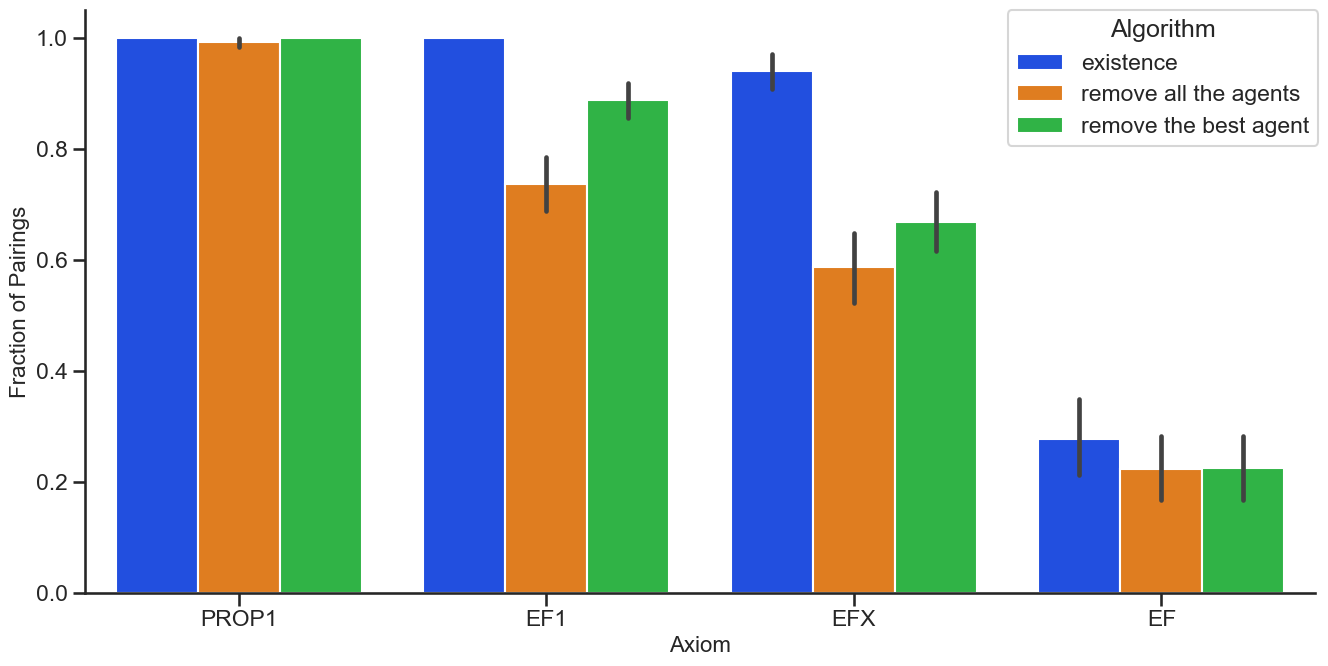}
        \caption{Only instances with $m \leq 5$ goods are considered, resulting in a total of 108 instances.}
        \label{fig:1a}
    \end{subfigure}
    \hfill
    \begin{subfigure}[b]{0.47\textwidth}
        \includegraphics[width=\textwidth]{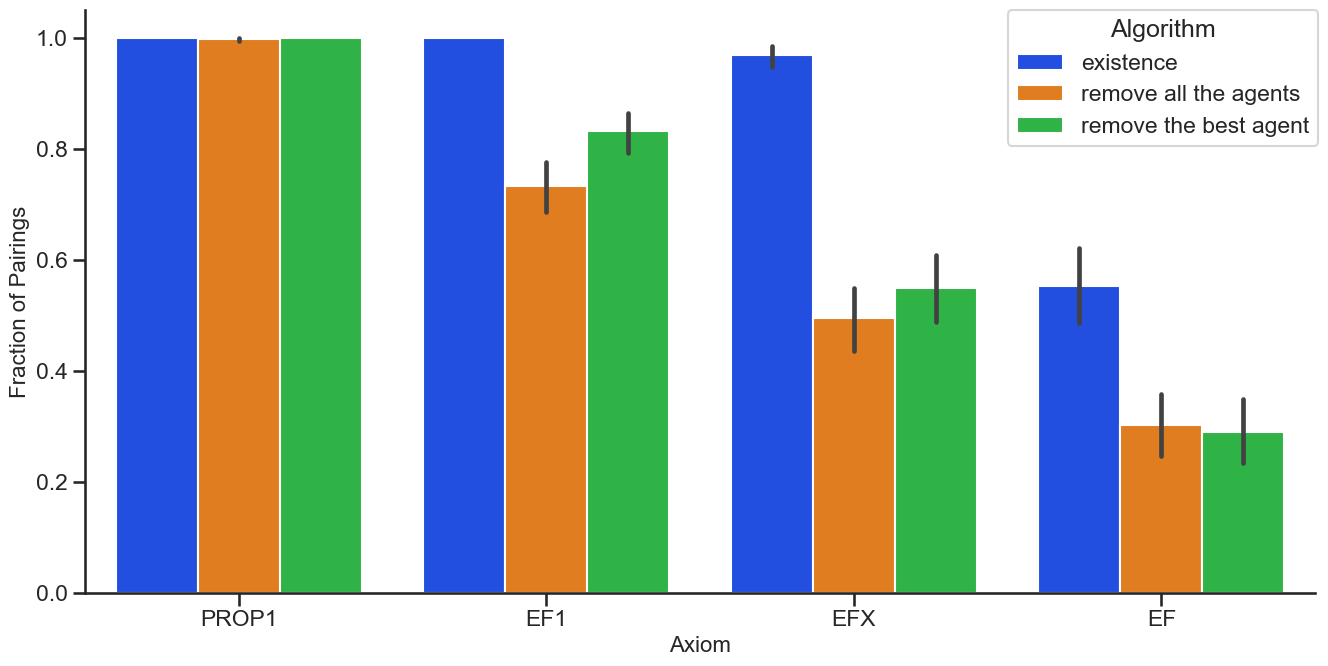}
        \caption{Only instances with $m \geq 6$ goods are considered, resulting in a total of 146 instances.}
        \label{fig:1b}
    \end{subfigure}

    \vspace{0.5cm}

    \begin{subfigure}[b]{0.47\textwidth}
        \includegraphics[width=\textwidth]{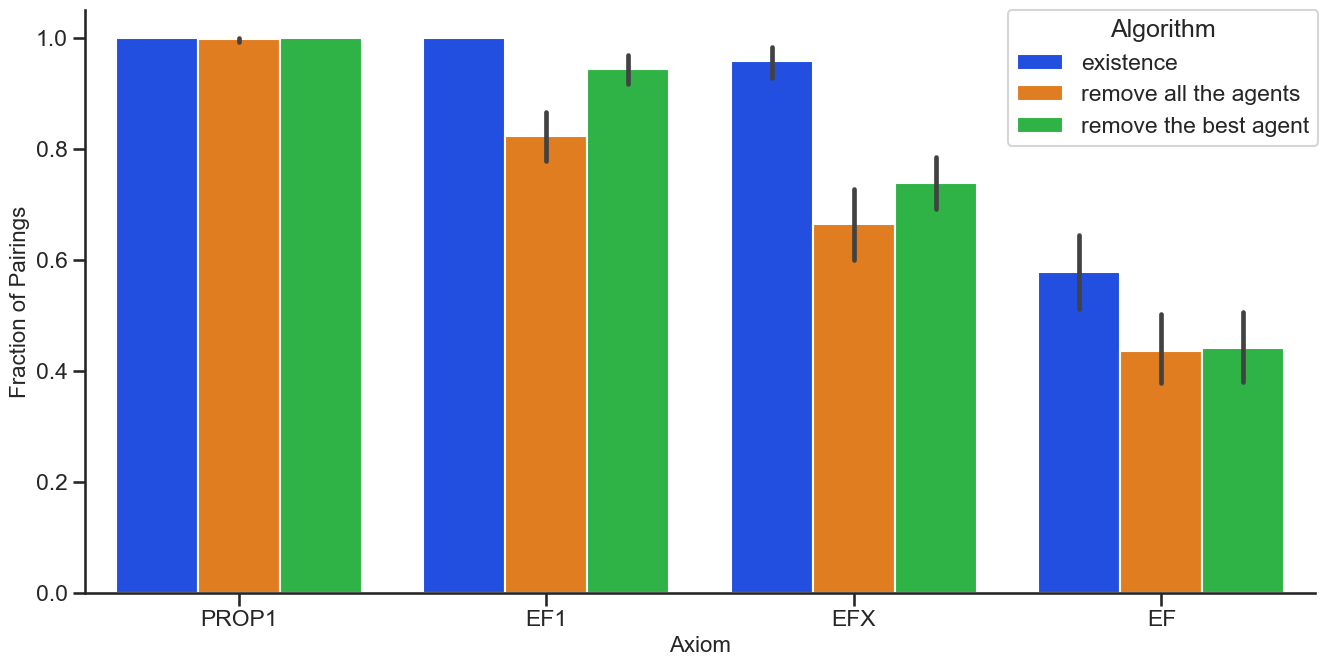}
        \caption{Only instances with $4$ agents are considered, resulting in a total of 120 instances.}
        \label{fig:1c}
    \end{subfigure}
    \hfill
    \begin{subfigure}[b]{0.47\textwidth}
        \includegraphics[width=\textwidth]{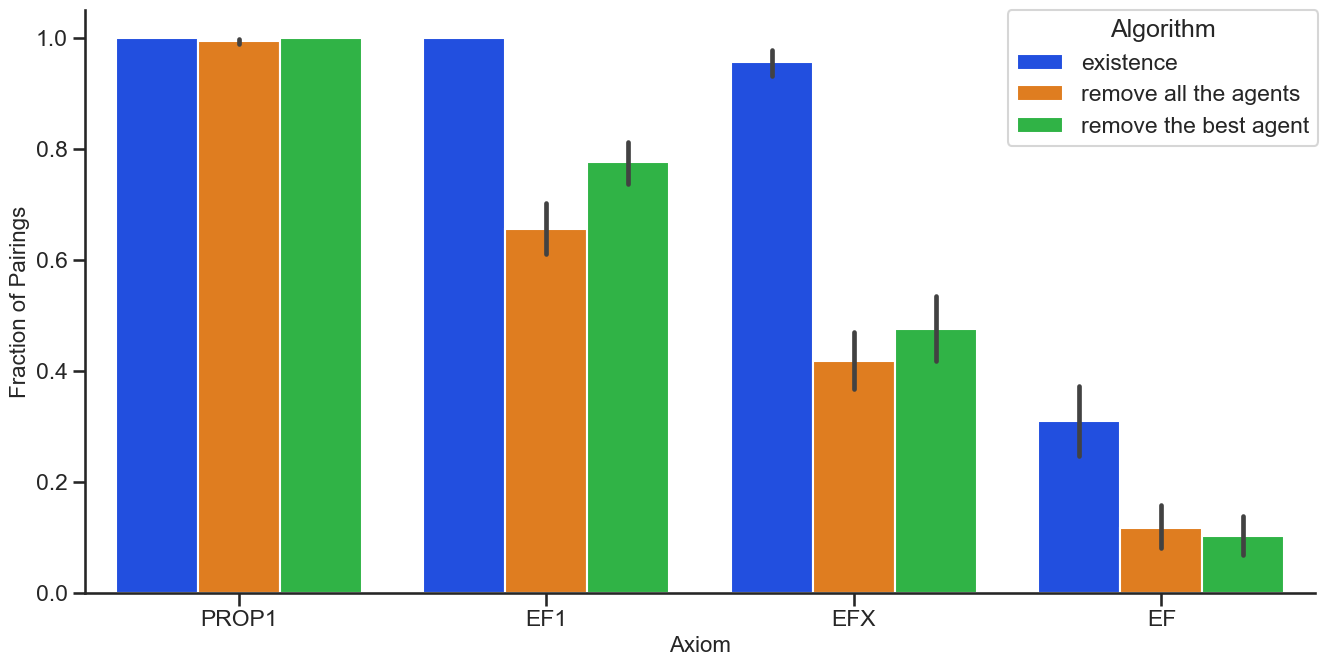}
        \caption{Only instances with $5$ or more agents are considered, resulting in a total of 134 instances.}
        \label{fig:1d}
    \end{subfigure}

    \caption{Fraction of pairings for which fair allocations exist or are found by one of two algorithms, averaged over all considered instances. Axioms imply axioms to their left. Error bars indicate 95\% confidence intervals (bootstrapping).}
    \label{fig:app:plots}
\end{figure}

\end{document}